\documentclass[USenglish, 11pt]{article}
\usepackage[T1]{fontenc}
\usepackage{amsthm,amsmath,amssymb,amsfonts,authblk,hyperref,setspace}
\usepackage{graphicx}
\usepackage{xcolor}
\usepackage{authblk}
\usepackage{enumitem}
  \usepackage{algorithm}
  \usepackage{algorithmic,float} 

\usepackage{cite}
\usepackage{stfloats}
\usepackage{url}

\usepackage{float}
\usepackage{tikz}

\newtheorem{theorem}{Theorem}[section]

\newtheorem{lemma}[theorem]{Lemma}

\newtheorem{proposition}[theorem]{Proposition}

\usepackage{fullpage}

\newcommand{\comment}[1]{}
\newcommand{\tildeO}{\widetilde{O}}

\newcommand{\cI}{\mathcal{I}}
\newcommand{\cJ}{\mathcal{J}}
\newcommand{\cS}{\mathcal{S}}
\newcommand{\cT}{\mathcal{T}}
\newcommand{\cB}{\mathcal{B}}
\newcommand{\cX}{\mathcal{X}}
\newcommand{\cY}{\mathcal{Y}}
\newcommand{\cK}{\mathcal{K}}
\newcommand{\cR}{\mathcal{R}}

\newcommand{\cQ}{\mathcal{Q}}
\newcommand{\cG}{\mathcal{G}}

\newcommand{\ED}{\textrm{SLOW-GED}}

\newcommand{\ed}{\Delta_\textrm{edit}}
\newcommand{\cost}{\textrm{cost}}
\newcommand{\displace}{\textrm{disp}}

\newcommand{\ACCEPT}{{\bf\textsc{accept}}}
\newcommand{\REJECT}{{\bf\textsc{reject}}}
\newcommand{\CLOSE}{{\bf\textsc{close}}}
\newcommand{\FAR}{{\bf\textsc{far}}}

\newcommand{\ZoomIn}{\textrm{ZoomIn}}

\newcommand{\Intervals}{\textrm{Intervals}}
\newcommand{\SparseSample}{\textrm{SparseSample}}
\newcommand{\Sparse}{\textrm{Sparse}}
\newcommand{\Bbelow}{\cB^\textrm{below}}
\newcommand{\Bdense}{\cB^\textrm{dense}}
\newcommand{\APM}{\textrm{APM}}
\newcommand{\Preprocess}{\textrm{Preprocess}}
\newcommand{\ProcessDense}{\textrm{ProcessDense}}
\newcommand{\Enumerate}{\textrm{EnumerateClose}}
\newcommand{\candidate}[2]{\langle #1; #2\rangle}

\newcommand{\Round}{\textrm{Round}}

\newcommand{\FAED}{\textrm{\bf{FAST-ED-UB}}}
\newcommand{\GAPED}{\textrm{\bf {GAP-ED}}}

\newcommand{\MAIN}{\textrm{FAST-GED}}

\newcommand{\ncost}{\textrm{ncost}}
\newcommand{\editd}{d_\textrm{edit}}

\newcommand{\Ht}{{\widetilde{H}}}

\newcommand{\startv}{{\textrm{source}}}
\newcommand{\finalv}{{\textrm{sink}}}

\newcommand{\qual}{q}
\newcommand{\epssub}[1]{{\varepsilon(#1)}}
\newcommand{\INDUCED}{\textrm{InducedBoxes}}
\newcommand{\winners}{\mathcal{M}}

\newcommand{\Gtt}{{\overline{G}}}
\newcommand{\SG}{\mathbf{SG}}
\newcommand{\SSE}{\mathbf{SS}}
\newcommand{\PD}{\mathbf{PD}}
\newcommand{\SR}{\mathbf{SR}}

\newcommand{\gapcondition}{\textrm{gap-condition}}
\newcommand{\pwrround}[1]{\lfloor #1 \rfloor_2}

\newcommand{\cmt}[1]{ $//$ {\em #1} $//$ }

\newcommand{\BSG}{B_\mathrm{SG}}
\newcommand{\BSS}{B_\mathrm{SS}}
\newcommand{\BPD}{B_\mathrm{PD}}

\begin{document}


\setcounter{page}{0}

\title{Constant factor approximations to  edit distance on far input pairs in nearly linear time}

\author[1]{Michal Kouck{\'{y}}\thanks{Email: koucky@iuuk.mff.cuni.cz. The research leading to these results has received funding from the European Research Council under the European Union's Seventh Framework Programme (FP/2007-2013)/ERC Grant Agreement no. 616787.
Partially supported by the Grant Agency of the Czech Republic under the grant agreement no. 19-27871X.}}
\author[2]{Michael Saks\thanks{Email: msaks30@gmail.com. Supported in part by Simons Foundation under award 332622.}}
\affil[1]{Computer Science Institute of Charles University,
Malostransk{\'e}  n{\'a}m\v{e}st\'{\i} 25,
118 00 Praha 1, Czech Republic}
\affil[2]{Department of Mathematics, Rutgers University, Piscataway, NJ, USA}


\date{}

\maketitle

\begin{abstract}
For any $T \geq 1$, there are constants $R=R(T) \geq 1$ and  $\zeta=\zeta(T)>0$ and a randomized algorithm
that takes as input an integer $n$ and  two strings $x,y$ of length at most $n$, and  runs in time $O(n^{1+\frac{1}{T}})$ and outputs an upper bound $U$
on the edit distance of $\editd(x,y)$ that with high probability,  satisfies $U \leq  R(\editd(x,y)+n^{1-\zeta})$.  In particular, on any input with
$\editd(x,y) \geq n^{1-\zeta}$ the algorithm outputs a constant factor approximation with high probability.  A similar result has been proven
independently by Brakensiek and Rubinstein~\cite{BR19}.
\end{abstract}



\thispagestyle{empty}
\newpage

\section{Introduction}
\label{sec:intro}
\noindent

The \emph{edit distance} (or \emph{Levenshtein distance})~\cite{Lev65} between strings $x,y$,
denoted by $\editd(x,y)$, is the minimum number of character insertions, deletions, and substitutions needed to convert $x$ into $y$. 
It was recently shown independently that edit distance can be approximated within a constant factor in truly subquadratic time in the quantum computation model~\cite{BEGHS18,BEGHS18A}.
and in the classical model \cite{CDGKS18,CDGKS-arxiv}.  The running time for a classical algorithm obtained
in \cite{CDGKS18,CDGKS-arxiv}  is $\widetilde{O}(n^{12/7})$, which was improved by Andoni \cite{And18} to   $O(n^{3/2+\epsilon})$.

This raises the natural question: what is the best possible running time of a constant factor approximation classical algorithm.
We make progress on this problem by developing a nearly linear time algorithm that gives a constant factor approximation when restricted to inputs whose edit distance is {\em not too small}:

\begin{theorem}\label{thm:main}
For every $T \geq 1$ there are constants $\zeta=\zeta(T)$ and $R=R(T)$ and a randomized algorithm $\FAED^{T}$ that takes as input an integer $n$ and
 two strings $x$ and $y$, with $|x|,|y| \leq n$,  over an (arbitrary) alphabet $\Sigma$, and runs in time 
$\widetilde{O}(n^{1+\frac{1}{T}})$ and outputs an upper bound  $U$ on $\editd(x,y)$, such that 
with probability at least $1-1/n$, $U \leq R\cdot (\editd(x,y)+n^{1-\zeta})$.
\end{theorem}

In particular, on any input $x,y$ with $\editd(x,y) \geq n^{1-\zeta}$ the algorithm gives a constant factor approximation.
The additive $n^{1-\zeta}$ term arises from some technical limitations in our algorithm and analysis, but since known algorithms for exact edit distance problem run faster on instances $x,y$  with small edit distance (\cite{LMS98,UKK85}) 
we  expect that it should be possible
to extend our result to give a nearly linear constant approximation algorithm for all ranges of edit distance. 

Brakensiak and Rubinstein~\cite{BR19} independently obtained essentially the same theorem.  While both our work and theirs builds on the techniques of~\cite{CDGKS18,CDGKS-arxiv,BEGHS18,BEGHS18A}, the algorithms have quite different structure.

\noindent
{\bf Other prior work}   (quoted from~\cite{CDGKS18}.)
Edit distance can be evaluated exactly in quadratic time  via
dynamic programming (Wagner and Fischer~\cite{WF74}). 
Masek and Paterson\cite{MP80} obtained the first (slightly) sub-quadratic  $O(n^2 /\log n)$ time algorithm, and the current asymptotically fastest algorithm (Grabowski~\cite{G16}) runs in time
 $O(n^2\log \log n /\log^2 n)$.    
Backurs and Indyk \cite{BI15} showed that a {\em truly sub-quadratic algorithm} ($O(n^{2-\delta})$ for some $\delta>0$) would imply a $2^{(1-\gamma)n}$ time algorithm
for CNF-satisfiabilty, contradicting the Strong Exponential Time Hypothesis (SETH). Abboud et al.~\cite{AHWW16} showed that even shaving an arbitrarily large polylog factor from $n^2$ would have the plausible, but apparently 
hard-to-prove, consequence that NEXP does not have non-uniform ${NC}^1$ circuits. For further ``barrier'' results, see~\cite{ABW15,BK15}. 

There is a long line of work on {\em approximating} edit distance.
The exact $O(n+k^2)$ time  algorithm (where $k$ is the edit distance of the input) of Landau {\em et al.}~\cite{LMS98} yields a  linear time $\sqrt{n}$-factor approximation.
This approximation factor was improved, first to $n^{3/7}$~\cite{BJKK04}, then to $n^{1/3+o(1)}$~\cite{BES06} and later to $2^{\widetilde{O}(\sqrt{\log n})}$~\cite{AO09}, all with slightly superlinear runtime. 
Batu {\em et al.}~\cite{BEKMRRS03} provided an $O(n^{1-\alpha})$-approximation algorithm with runtime $O(n^{\max\{\frac{\alpha}{2}, 2\alpha-1\}})$. The strongest result of this type is the $(\log n)^{O(1/\epsilon)}$ factor
approximation (for every $\epsilon>0$) with running time  $n^{1+\epsilon}$ of  Andoni {\em et al.}~\cite{AKO10}.
Abboud and Backurs~\cite{AB17} showed that a truly sub-quadratic deterministic time $1+o(1)$-factor approximation algorithm for  edit distance  would imply new circuit lower bounds. 

Andoni and Nguyen~\cite{AN10} found a randomized algorithm that approximates Ulam distance of two permutations of $\{1,\ldots,n\}$ (edit distance with only insertions and deletions) within a (large) constant factor in time $\tildeO(\sqrt{n}+n/k)$, where $k$ is the Ulam distance of the input; this was improved by Naumovitz {\em et al.}~\cite{NSS17} to a $(1+\varepsilon)$-factor approximation (for any  $\varepsilon>0$) with similar runtime. 


\noindent
{\bf Reduction to a Gap-Algorithm.}
For simplicity we will assume that the bound $n$ on the length $\max(|x|,|y|)$ is a power of  2 and $|x|=|y|=n$.  It is easy  to reduce the general
case to this case: on input $x',y'$, let $n$ be the least power of 2 that is at least $\max(|x'|,|y'|)$ and pad
both $x'$ and $y'$ using a single new symbol to obtain strings $x$, $y$ of length $n$.
It is easy to verify that $\editd(x',y') \leq \editd(x,y) \leq 2\editd(x',y')$, and so it suffices to approximate
$\editd(x,y)$. 

Following a common paradigm for approximation algorithms, our approximation algorithm is built by reducing to a {\em gap algorithm}.
In this paper, we consider randomized gap algorithms for edit distance. 
 These algorithms take as input
$(n,\theta,\delta;x,y)$   
where $n$ is an integral power of 2, $x$ and $y$ are strings of length $n$,
 $\theta \in (0 ,1]$ is a nonnegative power of 1/2 
and $\delta \in (0,1)$.
The triple $(n,\theta,\delta)$ are referred to as the {\em input parameters}
and $x,y$ as the {\em input strings}.   
We say that the algorithm has {\em quality $Q$ with respect to  $(n,\theta,\delta)$} provided that for all
strings $x,y$ of length $n$:
\begin{description}
\item[Gap Algorithm Soundness.]If $\editd(x,y) > Q \theta n$ then the algorithm returns \REJECT.
\item[Gap Algorithm Completeness.] If $\editd(x,y) \leq \theta n$ then the algorithm returns \ACCEPT{} with probability at least $1-\delta$.
\end{description}

We say that the algorithm {\em satisfies $\gapcondition(T,\zeta,Q)$}, where $T,Q \geq 1$ and $\zeta \geq 0$ provided that
for $n$ a power of 2,  and for all $\theta \geq n^{-\zeta}$
\begin{itemize}
\item The algorithm has quality $Q$ with respect to $(n,\theta,\delta)$,
\item The running time of the algorithm on any input $(n,\theta,\delta;x,y)$ is $\widetilde{O}(n^{1+1/T} \log(1/\delta))$ with probability 1.  Here $\widetilde{O}$
hides powers of $\log(n)$ whose exponent may depend on $T$.
\end{itemize}

We will prove:
\begin{theorem}
\label{thm:GUB}
For every $T \geq 1$ there are constants $\zeta=\zeta(T)>0$ and $Q=Q(T) \geq 1$, and a gap-algorithm
$\GAPED^{T}$ that satisfies $\gapcondition(T,\zeta,Q)$.
\end{theorem}

In Section~\ref{sec:main proof} we present the (routine) construction of the algorithm $\FAED^T$ from $\GAPED^T$, which proves Theorem~\ref{thm:main}.
The focus of the paper is on proving Theorem~\ref{thm:GUB}. 

\subsection{Speed-up routines}
\label{subsec:intro speed-up}  

Our algorithm, like that of~\cite{CDGKS18} is built from a {\em core speed-up algorithm} having access to  an existing
  "slow" gap algorithm. The speed-up algorithm produces a faster gap algorithm, with worse (but still constant) approximation quality,
while making queries to the slow algorithm on pairs of "short" substrings.
Given such a speed-up algorithm, one can build up a sequence of increasingly faster gap
algorithms $A_0,A_1,\ldots,$ where $A_0$ is just the quadratic exact edit distance algorithm, 
and $A_j$ is obtained by using the core  speed-up algorithm with $A_{j-1}$ playing the role of the "slow" algorithm.
If the  core speed-up algorithm involves some free parameters that may be optimized for best performance, this optimization
can be done separately for each $A_j$

The core speed-up algorithm designed in~\cite{CDGKS18}, gives an algorithm $A_1$ that has running time $\widetilde{O}(n^{12/7})$.  The algorithms
$A_j$ are successively faster, but do not get below $n^{\phi}$ where $\phi=1.61...$. 
The core speed-up algorithm we design in this paper gives a sequence of gap-algorithms where the
exponent of $n$ in the run-time converges to 1.

\section{Preliminaries}
\label{sec:preliminaries}

Many definitions and routine claims are adapted (with some modifications) from~\cite{CDGKS-arxiv}. 
The edit distance of  strings $u,v$ is denoted $\editd(u,v)$  and the {\em normalized edit distance} of $u,v$, denoted
$\ed(u,v)$ is defined to be $\editd(u,v)/|u|$. 

Throughout the paper $x$,$y$ denote two input strings of length $n$, where $n$ is a power of 2,  and $z$ denotes the concatenation $xy$. 
 
{\bf Intervals, Decompositions, aligned intervals, and $\delta$-aligned intervals.}
We consider intervals in $\{0,\ldots,2n\}$ which are as usual, subsets consisting of consecutive integers.  The {\em width} of interval $I$, $\mu(I)$ is equal to $\max(I)-\min(I)=|I|-1$.   Most intervals we consider  have width a power of 2. An interval of width $w$ is a $w$-interval.
Intervals index substrings of $z$, where 
$z_I$ denotes the substring indexed by the set $I\setminus\{\min(I)\}$, (Note that  $z_{\min(I)}$ is not part of $z_I$.  In particular, $z=z_{\{0,\ldots,2n\}}$, and $x=z_{\{0,\ldots,n\}}$ and $y=z_{\{n,\ldots,2n\}}$.

A  {\em decomposition} of an interval $I$ is a sequence $I_1, \ldots,I_k$ of intervals with $\min(I_1)=\min(I)$, $\max(I_k)=\max(I)$
and $\min(I_{j+1})=\max(I_j)$ for $j \in \{1,\ldots, k-1\}$.  Note that $z_{I_1},\ldots,z_{I_k}$ partitions the string $z_I$.

Let $w$ be a power of 2 that is at most $n$, and let $\delta$ be a power of 2 that is at most 1.
An interval of width $w$ is {\em aligned} if $\min(I)$ is a multiple of $w$ (and consequently $\max(I)$ is also a multiple of $w$).  
The interval is {\em $\delta$-aligned} if $\min(I)$ is a multiple of $\max(\delta w,1)$ (and consequently so is $\max(I)$).  In particular a 1-aligned interval
is  aligned.  
We define:
\begin{itemize}
\item $\Intervals(w)$ is the set of aligned intervals of width $w$, subsets of $\{0,\dots,n\}$.  
\item $\Intervals(w,\delta)$ to be the set of $\delta$-aligned intervals of width $w$, subsets of $\{0,\dots,2n\}$. 
\item For an interval $I$,  $\Intervals(w;I)=\{I' \in \Intervals(w):I' \subseteq I\}$,
and $\Intervals(w,\delta;I)=\{I' \in \Intervals(w,\delta):I' \subseteq I\}$.
\end{itemize}

Since $n$ and $w$ are powers of 2,
$\Intervals(w)$ is a decomposition of $\{0,\ldots,n\}$.  When we use the notation $\Intervals(w;I)$, $I$ will be an aligned interval
of width a power of 2, so that $\Intervals(w;I)$ is a decomposition of $I$.

{\bf The grid $\{0,\ldots,n\} \times \{0,\ldots,2n\}$, boxes and stacks.}
Consider the grid $\{0,\ldots, n\} \times \{0,\ldots,2n\}$ lying in the coordinate  plane.
For $S \subseteq \{0,\ldots,n\} \times \{0,\ldots,n\}$, the horizontal projection $\pi_H(S)$ is the set of first coordinates of elements of $S$,
and the vertical projection of $S$, $\pi_V(S)$ is the set of second coordinates. 
 
A {\em box} is a set $I\times J \subseteq \{0,\ldots,n\} \times \{0,\ldots,2n\}$ for intervals $I,J$, and it represents 
the pair $x_I,z_J$ of substrings.   Since $I \subseteq \{0,\ldots,n\}$, $z_I=x_I$.  Note that if $J \subseteq \{0,\ldots,n\}$
then $z_I,z_J$ is a pair of substrings of $x$ and if $J \subseteq \{n,\ldots,2n\}$, it is a pair (substring of $x$, substring of $y$).
$I \times J$ is a {\em $w$-box} if $\mu(I)=\mu(J)=w$. 
The lower left hand corner is $(\min(I),\min(J))$ and  the upper right hand corner is $(\max(I),\max(J))$.
Note that $\pi_H(I \times J)=I$ and $\pi_V(I \times J)=J$.
Box $I \times J$ is {\em horizontally aligned} if  $I$ is aligned, and it is  {\em vertically $\delta$-aligned} or simply
{\em $\delta$-aligned} if $J$ is
$\delta$-aligned; we have no need to refer to {\em horizontally $\delta$-aligned} boxes. Box $I \times J$ is {\em square} if $\mu(I)=\mu(J)$.

A {\em stack} is a set of boxes all having the same horizontal projection.  For interval $I$ and set of intervals $\mathcal{J}$,
$I \times \mathcal{J}$ is the stack $\{I \times J:J \in \cJ\}$.


{\bf Grid graphs}.
The {\em grid graph of $z$}, $G_z$, is a directed graph
 with edge costs, having
vertex set $\{0,\ldots,n\} \times \{0,\ldots,2n\}$   and 
 all edges of the form  $(i-1,j)\to (i,j)$ ($H$-edges), $(i,j-1) \to (i,j)$ ($V$-edges)
and $(i-1,j-1) \to (i,j)$ ($D$-edges).  Every H-edge  and V-edge costs 1, and  a D-edge
has cost 1 if $z_i \neq z_j$ and 0 otherwise.  $G_z$ is
acyclic, with edges moving "up and to the right".  A directed path $\tau$
joins a pair of vertices $\startv(\tau)$ and $\finalv(\tau)$ with $\startv(\tau) \leq \finalv(\tau)$.
The {\em box spanned by $\tau$} is the unique minimal box $I \times J$ that contains $\tau$;
this is equal to $\pi_H(\tau) \times \pi_V(\tau)$.  
We say  $\tau$ {\em traverses} $I \times J$ if
$I \times J$ is the box spanned by $\tau$, which is equivalent  to $\startv(\tau)=(\min(I),\min(J))$ and $\finalv(\tau)=(\max(I),\max(J))$.
A {\em traversal} of $I \times J$ is any path that traverses $I \times J$.

For $I \subseteq \pi_H(\tau)$, let $\tau_I$ denote the minimal subpath of $\tau$ whose horizontal projection is $I$.

\noindent
{\bf Cost and normalized cost.}
The {\em cost} of a directed path $\tau$, $\cost(\tau)$ is the sum of the edge costs, and the {\em normalized cost} is 
$\ncost(\tau)=\frac{\cost(\tau)}{\mu(\pi_H(\tau))}$.   The cost of box $I\times J$,  
$\cost(I\times J)$, is the min-cost of a traversal of $I \times J$
and  $\ncost(I\times J)=\frac{1}{\mu(I)}\cost(I \times J)$.  

It is well known (and easy to see) that for any box $I \times J$, a traversal of $I \times J$ corresponds to an {\em alignment} from $a=z_I$ to $b=z_J$, i.e.
a set of character deletions, insertions and substitutions that changes $a$ to $b$, where
an H-edge  $(i-1,j)\to (i,j)$ corresponds to "delete $a_i$", a V-edge  $(i,j-1)\to (i,j)$
corresponds to "insert  $b_j$ between $a_i$ and $a_{i+1}$" and a D-edge $(i-1,j-1)\to (i,j)$  corresponds to
replace $a_i$ by $b_j$, unless they are already equal.  Thus:

\begin{proposition}
\label{prop:def}
The cost of an alignment corresponding to path $\tau$ is $\cost(\tau)$.  Thus for any $I, J \subseteq \{0,\ldots,2n\}$,
$\editd(z_I,z_J)=\cost(I \times J)$.  In particular $\editd(x,y)=\cost(\{0,\ldots,n\} \times \{n,\ldots,2n\})$.
\end{proposition}

\medskip
{\bf Displacement of a box relative to a path or box.}
The following easy fact (noted in~\cite{CDGKS18}) relates the cost of two boxes having the same horizontal projection:

\begin{proposition} 
\label{prop:sym diff}
For intervals $I,J,J' \subseteq \{0,\ldots,n\}$,
$|\cost(I \times J)-\cost(I \times J')| \leq |J \Delta J'|$, where $\Delta$ denotes symmetric difference.
\end{proposition}

 Let $\tau$ be a path whose horizontal projection includes $I$.
The {\em displacement} of the square  box $I \times J$ with respect to $\tau$, $\displace(I \times J,\tau)$  is the smallest $K$  such
that $(\min(I),\min(J))$ is within $K$ vertical units of $\startv(\tau_I)$ and $(\max(I),\max(J))$ is within $K$ vertical units of $\finalv(\tau_I)$.

We make a few easy observations.

\begin{proposition}
\label{prop: path displacement}
Let $\tau$ be a path whose horizontal projection includes $I$ and let $I \times J$ be a box.
Then $\cost(I \times J) \leq \cost(\tau_I)+2\displace(I \times J,\tau)$.
\end{proposition}

\begin{proof}
Let $J'$ be the vertical projection of $\tau_I$.  Then:
$\cost(I \times J) \leq \cost(I \times J')+|J \Delta J'| \leq \cost(\tau_I) +|J \Delta J'| \leq \cost(\tau_I) + 2\displace(I \times J,\tau)$.
\end{proof}

The following fact (which is essentially the same as Proposition 3.4 of~\cite{CDGKS-arxiv}) says that every path $\tau$ with projection $I'$ can be approximately covered by a $\delta$-aligned box whose cost is close
to $\cost(\tau)$ and whose displacement from $\tau$ is small:

\begin{proposition}
\label{prop:align-box}
Let $I'$ and $J$ be intervals and suppose $\delta \in (0,1]$.
Let $\tau$ be a path lying inside of $I' \times J$
whose horizontal projection  is $I'$.
There is a $\delta$-aligned interval $J'$ of width $\mu(I')$ so that 
$\displace(I' \times J',\tau_{I'}) \leq \delta \mu(I') + \cost(\tau_{I'})$ and
$\ncost(I' \times J') \leq 2\ncost(\tau_{I'})+\delta$.
\end{proposition}

\begin{proof}
Let $J$ be the vertical projection of $\tau_{I'}$. 
If $\mu(J) \ge \mu(I')$ then let $\hat{J}$ be the interval of width $\mu(I')$ with $\min(\hat{J})=\min(J)$.
Otherwise let $\hat{J}$ be any interval of width $\mu(I')$ that contains $J$.

The box $I' \times \hat{J}$ has displacement at most $\cost(\tau_{I'})$ from $\tau_{I'}$, and
has cost at most $2\cost(\tau_{I'})$. Finally, let $J'$ be obtained by shifting $\hat{J}$ up or down to the closest
$\delta$-aligned interval. This shift is at most $\delta/2$ units.
This increases both the displacement and the cost by at most $\delta \mu(I')$.
\end{proof}

The {\em diagonal} of a square box $I \times J$ is the diagonal path joining $(\min(I),\min(J))$ to $(\max(I),\max(J))$.
Let $I \times J$ and $I' \times J'$ be square boxes with $I' \subseteq I$.  The {\em displacement} of $I' \times J'$ with respect to $I \times J$,
$\displace(I' \times J',I \times J)$
is the displacement of $I' \times J'$ with respect to the diagonal of $I \times J$, which is just the number of vertical units
one needs to shift $I' \times J'$ so that its diagonal is a subpath of the diagonal of $I \times J$.

\begin{proposition}
\label{prop:point displacement}
Suppose $\tau$ traverses the square box $I \times J$ of width $w$.  Then every point of $\tau$ is within vertical distance
$\cost(\tau)/2$ of the diagonal of $I \times J$.
\end{proposition}

\begin{proof}
Consider a point of $\tau$ expressed as $P=(\min(I)+u,\min(J)+v)$.   Then $\tau$ can be split into
two parts $\tau_1$, ending at $P$ and $\tau_2$ starting at $P$. Then $\cost(\tau)=\cost(\tau_1)+\cost(\tau_2)  \geq 2|v-u|$
which is twice the vertical distance of $P$ to the diagonal of $I \times J$.  
\end{proof}

\medskip
{\bf Weighted boxes and stacks, certified boxes and stacks, shortcut graphs.}

A {\em weighted box} is a pair $(I \times J,\kappa)$ where $\kappa \geq 0$. If $\ncost(I \times J) \leq \kappa$
we say that  $(I \times J,\kappa)$  is a {\em certified box}.
A {\em weighted stack} $(I \times \mathcal{J},\kappa)$ is a pair where $I \times \cJ$ is a stack and $\kappa\geq 0$.
We associate $(I \times \mathcal{J},\kappa)$ with the
set $\{(I \times J,\kappa):J \in \mathcal{J}\}$.
If every box in $(I \times \cJ,\kappa)$ is certified, we call it a {\em certified stack}.  

Let $\widetilde{G}$ be the digraph on $\{0,\ldots,n\} \times \{0,\ldots, 2n\}$ with arc set
$\{(i,j) \rightarrow (i',j'): i \leq i', j \leq j',(i,j) \neq (i',j')\}$  The edges
with $i<i'$ and $j<j'$ are called {\em shortcuts}.  Associated to any weighted box $(I \times J,\kappa)$ 
there is a weighted shortcut edge $(\min(I),\min(J)) \rightarrow (\max(I),\max(J))$ with weight $\kappa \mu(I)$.
Given a set $\cR$ of weighted boxes, we define the {\em weighted shortcut graph} $\widetilde{G}(\cR)$
to be the weighted directed graph consisting of all H-edges and V-edges with weight 1, and
all of the shortcut edges corresponding to the boxes in $\cR$.   For a box
$I \times J$, let $\cost_{\cR}(I \times J)$ denote the minimum cost of a traversal
of $I \times J$ in $\widetilde{G}(\cR)$.

 If every box in $\cR$ is certified
we say that $\widetilde{G}(\cR)$ is a {\em certified shortcut graph}.
 A certified shortcut graph $\bar{G}(\cR)$
provides upper bounds on the edit distance.
We omit the proof of the following easy fact:

\begin{proposition}
\label{prop:shortcut}
Let $\cR$ be a set of certified boxes.  For any box $I \times J$,
$\editd(z_I,z_J) \leq \cost_{\cR}(I \times J)$.
\end{proposition}

\section{The core speed-up algorithm of~\cite{CDGKS18}}
\label{sec:cdgks}

As discussed in Section~\ref{subsec:intro speed-up}, the main ingredient in~\cite{CDGKS18} is a core speed-up algorithm
that has access to a slow edit distance approximation algorithm and uses it to build a faster approximation algorithm. 
We review the main ideas of the core speed-up algorithm in~\cite{CDGKS18}, which provides the starting point for ours. 
To simplify the description 
we assume that the slow edit distance algorithm is just the quadratic exact edit distance algorithm.
In their work, they reduce to the case $\theta > n^{-1/5}$ and build a subquadratic time algorithm for the
gap-problem where $\theta \geq  n^{-1/5}$.
The algorithm operates in two phases.
The {\em discovery phase}  generates a set $\cQ$ of certified boxes.   In the {\em shortest path phase} the
algorithm evaluates the cost of $(\{0,\ldots,n\} \times \{n,\ldots,2n\})$ in the shortcut graph $\widetilde{G}(\cR)$ where $\cR$ is a set of
certified boxes obtained by a minor modification of $\cQ$.  Proposition~\ref{prop:shortcut} implies that this
is an upper bound on $\editd(x,y)$. The main work is to define the discovery phase to ensure that this upper bound is not too much bigger than
the true value.    The shortest path phase is implemented by a straightforward variant of dynamic programming.

The discovery phase is
defined in terms of parameters $w_1<d< w_2$, which are powers of 2 that are, respectively, approximately $n^{1/7}$, $n^{2/7}$ and $n^{3/7}$.  
The set $\cQ$ consists of certified $w_1$-boxes and certified $w_2$-boxes, and satisfies with high probability: for
every horizontally
aligned $w_2$-box $I \times J$, $\cost_{\cR}(I \times J)\leq C \cdot [ \cost(I \times J)+\theta w_2 ]$ for some constant $C$.  It is not difficult 
to show that this implies that the upper bound on $\editd(x,y)$ output by the shortest path inference phase will be at most $C \cdot [\editd(x,y)+\theta n]$, 
which is enough to solve the gap-problem.

The algorithm  generates boxes of width $w_1$ iteratively for $i$ from $0,\ldots,\log(1/\theta)$ and $\epssub{i}=2^{-i}$. 
For each horizontally aligned $I$, let $\mathcal{N}_{\epssub{i}}(I)$ be the set of $J$ that are $\epssub{i+3}$-aligned and satisfy
$\ncost(I \times J) \leq \epssub{i}$.
Iteration $i$ starts by classifying each of the $n/w_1$-aligned $w_1$-intervals,  as {\em dense}  or {\em sparse} subject to the requirement that
every $I$ with $N_{\epssub{i}}(I) \geq 2d$ is classified as dense, and every $I$ with $N_{\epssub{i}}(I) \leq d/2$ is classified as sparse;
this classification of $I$ is
done with high probability by sampling $J$ at a rate $\log(n)/d$ and calling $I$ dense (resp. sparse) if at least (resp. at most) $\log(n)$ of the sample
are within distance $\epssub{i}$ of $I$.
Next for each dense interval $I$ a set  $\mathcal{J}(I)$ of $\epssub{i+3}$-aligned $w_1$-intervals $J$
is constructed such that $\ncost(I\times J) \le 5 \epssub{i}$ and $\mathcal{N}_{\epssub{i}}(I) \subseteq \mathcal{J}(I)$.
For any given $I$ we can construct $\mathcal{J}(I)$ by computing its edit distance with every $\epssub{i}/8$-aligned interval,
in time $O(n w_1/\epssub{i})$.  If we do this for all $n/w_1$-aligned intervals the time is
 $\Theta(n^2/\epssub{i})$, but the restriction to dense intervals allows a savings of a factor of $\epssub{i} d$: 
Initialize $\mathcal{D}$ to be the set of dense aligned  $w_1$-intervals.  While $\mathcal{D} \neq \emptyset$ choose $I \in \mathcal{D}$ (the {\em pivot} for the current round) and construct $\mathcal{X} = N_{2\epssub{i}}(I)$ and $\mathcal{Y} = N_{3\epssub{i}}(I)$ and certify
all boxes $(I' \times J',5\epssub{i})$ for $I' \in \mathcal{X}$ and $J' \in \mathcal{Y}$. Delete $\mathcal{X}$ from $\mathcal{D}$
and continue.   The number of pivots is thus only
 $O(n/w_1\epssub{i}d)$ since the sets $N_{\epssub{i}}(I)$ are of size at least $d$ and 
are disjoint for different pivots.

The rest of the discovery phase constructs a (relatively small) set of $w_2$-boxes.     For each horizontally aligned $w_2$-interval $I'$,  the
$w_1$-subintervals of $I'$ that were declared sparse (over all iterations of $i$)  are used to select a small subset $\mathcal{J}'(I')$ of 
the $w_2$-intervals, and we certify each box $I' \times J'$ for $J' \in \mathcal{J}'(I')$ by computing their
edit distance exactly.  The set $\mathcal{J}'(I')$ is obtained as follows: For each $i \in \{0,\ldots,\log (1/\theta)\}$,
select a polylog$(n)$ size subset $\mathcal{S}_i(I')$ of the subintervals of $I'$ that were declared sparse in iteration $i$,
and for each $I'' \in \mathcal{S}_i(I')$ exactly compute $\cost(I'',J)$ for all $\epssub{i+3}$-aligned intervals $J$
to determine $\mathcal{N}_{\epssub{i}}(I'')$ (which has size at most $2d$).  For each box $I'' \times J$, let $J'$ be the unique $w_2$-interval
such that the diagonal of $I'' \times J$ is a subset of the diagonal of   $I' \times J'$ and add $J'$ to $\mathcal{J}'(I')$.  The size of
$\mathcal{J}'(I')$ is $\widetilde{O}(d)$ and so the total cost of evaluating the edit distance of boxes $I' \times J'$ for
$I' \in \Intervals(w_2;\{0,\ldots,n\})$ and $J' \in \mathcal{J}'(I')$ is $\widetilde{O}(n d w_2)$.  

The parameters $w_1,d,w_2$ are adjusted to minimize the run time at $\widetilde{O}(n^{12/7})$.
The key claim in~\cite{CDGKS18} is that for every horizontally aligned $w_2$-box  $I \times J$, the boxes from the discovery phase
imply an upper bound $\ncost(I \times J)$ that is at most $C\cdot \ncost(I \times J) + C'\theta$ which is sufficient for the shortcut phase to succeed.  The claim is proved by showing
that  if the set of certified $w_1$-boxes does not imply a sufficiently good upper bound on $\ncost(I \times J)$, then with high probability, one of the $w_2$-boxes $I \times J'$ constructed
in the second part of the discovery phase is within a small vertical shift of $I \times J$, and therefore can be used in the inference phase
to imply a good upper bound on $\cost(I \times J)$.   

\section{The new core speed-up algorithm}

The main new ingredient of the new core speed-up algorithm presented here  is the replacement of the pair $w_1<w_2$ of widths from~\cite{CDGKS18} by a hierarchy 
$w_1<\cdots<w_k$ of widths.  While the idea of such an extension is natural, it is not a priori clear how to extend the ideas of~\cite{CDGKS18} to such a hierarchy.   Our new algorithm proceeds in  $k$ iterations.  During iteration $j$ the algorithm
builds a data structure that supports approximate distance queries between substrings of width $w_j$.
Each successive data structure recursively uses the data structure from the previous iterations.  Iteration $j$  is accomplished by
 a suitable variant of the algorithm from~\cite{CDGKS18}.  

The algorithm of~\cite{CDGKS18}  splits neatly into a discovery phase and an inference phase.  In the new algorithm, each iteration
begins with an inference phase (using boxes discovered in the previous phase) followed by a discovery phase.

Here is our main speed-up theorem.

\begin{theorem}
\label{thm:speedup}  Suppose that $\ED$ is a gap algorithm for edit distance satisfying $\gapcondition(T',\zeta',Q')$ where $T' \geq 1$,
$\zeta'>0$ and $Q' \geq 1$. There is an algorithm $\MAIN$
(using $\ED$ as a subroutine) that satisfies $\gapcondition(T,\zeta,Q)$ with $T=T'+1/6$ where 
 $\zeta>0$ and $Q\geq 1$ are suitably chosen (depending only on $T'$,$\zeta'$ and $Q'$).
\end{theorem}

Applying this theorem inductively with $A_0$ being the exact edit distance algorithm, we get a sequence of algorithms
$A_j$ where $A_j$ satisfies $\gapcondition(1+j/6, \zeta_j,Q_k)$ for suitable constants $\zeta_j>0$ and $Q_j$,
and taking $j=6(T-1)$ gives 
Theorem~\ref{thm:GUB}.

The proof of Theorem~\ref{thm:speedup} is the heart of the paper.   
We describe the algorithm in the following order:
\begin{enumerate}
\item The parameters used by the algorithm (Section~\ref{subsec:params}).
\item The overall architecture, including data objects, of the algorithm  (Section~\ref{subsec:objects}).
\item Some basic functions used in the algorithm (Section~\ref{subsec:primitives}).
\item The mechanics of the algorithm. (Section~\ref{subsec:mechanics}).
\item The use of randomness in the algorithm  (Section~\ref{subsec:randomness}).
\item The properties enforced by the algorithm (Section~\ref{subsec:properties} and~\ref{subsec:completeness}).
\item The proof that \MAIN{} satisifes the gap-algorithm Soundness and Completeness requirements  (Section~\ref{subsec:main proof}).
\item The running time analysis in terms of the parameters (Section~\ref{subsec:time analysis}).
\item The choice of parameters that attain the run time claims for \MAIN{} (Section~\ref{subsec:parameter choice}).
\item Tying up the proof of Theorem~\ref{thm:speedup} (Section~\ref{subsec:tie up}).
\end{enumerate}

\subsection{The algorithm parameters}
\label{subsec:params}

Recall that a gap-algorithm takes as input $(n,\theta,\delta;x,y)$ where $n$ is a power of 2 and $|x|=|y|=n$, and $\theta \in (0,1]$ is a power of 1/2.

In our description of the algorithm,
we fix the input parameter $\delta$ in the algorithm \MAIN{}  to $\delta=1/2$.  
For $\delta<1/2$, we execute the algorihm with $\delta=1/2$ independently for
 $r=\lceil \log_2 (1/\delta) \rceil$ times, and \REJECT{} only if every run returns \REJECT.  This compound algorithm will \REJECT{} every input $x,y$ such that $\editd(x,y) \geq Q \theta n$, since every run will \REJECT.  The probability that the compound algorithm incorrectly returns
\REJECT{} on input with $\editd(x,y) \leq \theta n$ is at most $(1/2)^r \leq \delta$, as required.  

Second, we fix the value of $\delta$ for all calls of $\ED$ within \MAIN, to $\delta = n^{-12}$ where $n$ is the length of the global input
to \MAIN{}.  Since the number of calls to $\ED$ will be bounded above (easily) by $\tildeO(n^2)$, a union bound implies that the probability that  every
call to $\ED$ is correct is at least $1-n^{-8}$.

The algorithm $\MAIN$ takes as input  $n,\theta;x,y$ where
$n$ is a power of 2, $x$ and $y$ are strings of length $n$ and the gap parameter
$\theta \in (0,1]$ is a power of 1/2.   The algorithm sets
$z$ to be the concatenation of $xy$ and treats $z$ as a global variable.

The number of iterations (levels)  of $\MAIN$ is a parameter $k$.
For each $j \in {1,\ldots,k+1}$, there is a {\em width parameter} $w_j$ and for each $j \in \{0,\ldots,k\}$, there is a {\em density parameter} $d_j$.  These parameters are integer powers of 2 satisfying:\footnote{We denote by $\pwrround{\cdot}$ the closest power of two of size smaller or equal.}

$$w_1=\pwrround{\sqrt{n}}<w_2< \cdots < w_k < w_{k+1}=n.$$

$$d_0=\pwrround{\sqrt{n}} > d_1 > \cdots > d_k=1.$$

Furthermore, for $1 \leq j \leq k$:
\begin{eqnarray}
\label{eqn:tech assump2}
\frac{n}{w_j} &\ge& d_j.
\end{eqnarray}

These parameters will be chosen in Section~\ref{subsec:parameter choice} to optimize
the time analysis.  
For now we note a  technical assumption, that will be verified in Section~\ref{subsec:parameter choice},  that is needed in the analysis.  
 For $1 \leq j \leq k$:

\begin{eqnarray}
\label{eqn:tech assump}
\frac{w_j}{w_{j+1}} &\leq &\theta/2.
\end{eqnarray}

%

For each $j \in \{0,\ldots,k\}$, there are quality parameters $\qual_j$ that satisfy the recurrence:
\begin{eqnarray*}
\qual_0 &=&\log(Q') \hspace{.5 in}\text{(where $Q'$ is the quality  of $\ED$)}\\
\qual_j &=&3\qual_{j-1}+21 \hspace{.5 in} \text{for $j>1$}.
\end{eqnarray*}
The quality of the final approximation is $Q=2^{\qual_k+6}$

We also define, for integers $i$, $\epssub{i}=2^{-i}$.  In most cases, $i \in
\{0,\ldots, \log(1/\theta)\}$ so  
$1 \geq \epssub{i} \geq \theta$.

There is a constant $c_0$ used in the definition of the procedure $\ProcessDense$. (See Section~\ref{subsec:randomness}.)

\subsection{The  architecture of the algorithm, and the neighborhood data structure}
\label{subsec:objects}

\MAIN{} consists of $k$ iterations ({\em levels}), and a final post-processing step.   
 During iteration $j$, the algorithm examines pairs $\candidate{i}{I \times J}$, called {\em candidates},  where $i \in \{0,\ldots,\log(1/\theta)\}$,
$I \in \Intervals(w_j)$ and $J \in \Intervals(w_j,\epssub{i+3})$. (Hence, a candidate is any $\candidate{i}{I \times J}$ that satisfies some weak consistency requirements.)
The pair $I \times J$ is called a {\em level $j$ box}
and $\candidate{i}{I \times J}$ is a  {\em level $j$ candidate}.
Iteration $j$ implicitly classifies all level $j$-candidates as \CLOSE{} or \FAR.  This classification satisfies:

\begin{itemize}
\item If $\ncost(I \times J) \leq \epssub{i}$ then $\candidate{i}{I \times J}$ is classified as \CLOSE.
\item If $\ncost(I \times J) > \epssub{i-\qual_{j-1}-6}$ then $\candidate{i}{I \times J}$ is classified as \FAR.
\end{itemize}

If $\epssub{i}< \ncost(I \times J) \leq \epssub{i-\qual_{j-1}-6}$ then $\candidate{i}{I \times J}$ may be classified as either \CLOSE{} or \FAR.

This implicit classification is accomplished by
 a data structure, called the {\em neighborhood data structure}. The
data structure implements a query $\Enumerate$
which
takes as input $(j,I\times \cJ,i)$ where:
\begin{itemize}
\item  $j \in \{1,\ldots,k\}$ is the level,
\item  $I \times \mathcal{J}$ is a stack satisfying $I \in \Intervals(w_j)$ and $\mathcal{J} \subseteq \Intervals(w_j,\epssub{i+3})$,
\item $i \in \{0,\ldots,\log(1/\theta)\}$,  
\end{itemize}
and returns the set of $J \in \mathcal{J}$ for which $\candidate{i}{I \times J}$ is \CLOSE. 
In particular, $\Enumerate(j,I \times \{J\},i)$ returns $\{J\}$ if $\candidate{i}{I \times J}$ is \CLOSE{} and returns $\emptyset$ otherwise.
The pair $\candidate{i}{I \times \cJ}$ is called a {\em level $j$ candidate stack}.

The queries with level parameter $j$ are the
{\em level $j$ queries}.    Initially the data structure is unable to answer any queries.  During iteration $j$
 the algorithm constructs the part of the data structure that determines the classification of level $j$ candidates
as \CLOSE{} or \FAR, and thereby enabling level $j$ queries.  

At the start of iteration $j$, queries up to level $j-1$ have been enabled.
To enable $\Enumerate(j,\cdot)$ the algorithm constructs families of sets
for each $I \in \Intervals(w_j)$ and each $i \in \{0,\ldots, \log(1/\theta)\}$ as follows:

\begin{itemize}
\item 
A subset of $\Intervals(w_j,\epssub{i+3})$
denoted ${\bf \Bbelow}(j,I,i)$.       
\item A subset of
$\Intervals(w_{j-1};I)$ denoted
 ${\bf \SparseSample}(j,I,i)$.
\end{itemize}

The query $\Enumerate(j,\cdot)$ uses these sets, as well as calls to $\Enumerate(j-1,\cdot)$.
Thus the level $j$ neighborhood data structure consists of all of the sets $\Bbelow(j',\cdot)$ and $\SparseSample(j',\cdot)$ for
$1 \leq j' \leq j$.  

During iteration $j$, subroutines $\Preprocess$ and $\ProcessDense$ are called with parameter $j$.
The purpose of $\Preprocess(j)$ is to create the sets
 $\Bbelow(j,\cdot)$ and $\SparseSample(j,\cdot)$. The construction of these sets
 involves some random choices, which affect the \CLOSE/\FAR{} classification; but
once the choices are made the \CLOSE/\FAR{} classification is fixed.
 The creation of these sets activates 
$\Enumerate(j,\cdot)$.   While the data structure grows during each iteration
 to enable higher level queries, once $\Enumerate(j,\cdot)$ is enabled, the portion of the data structure
used to handle level $j$ queries is static.

The other procedure in iteration $j$ of \MAIN() is $\ProcessDense(j)$.   $\ProcessDense(j)$ creates
the following sets
 for each $i \in \{0,\ldots,\log(1/\theta)\}$:
\begin{itemize}
 \item  $\Sparse(j,i) \subseteq \Intervals(w_j)$.
\item For each $I  \not\in \Sparse(j,i)$, 
a subset of $\Intervals(w_j,\epssub{i+3})$
denoted {\bf $\Bdense(j,I,i)$}.
\item A set {\bf $\cR(j)$} of weighted boxes (which we will prove are all certified).
\end{itemize}

The  sets $\Bdense(j,\cdot)$ are local variables within $\ProcessDense(j)$, used 
to create $\cR(j)$.

The set $\cR(j)$ and $\Sparse(j,\cdot)$ are global variables but,
with the exception of the final iteration $j=k$,
they are used only in $\Preprocess(j+1)$, and then never used again.
Following iteration $k$, the set $\cR(k)$ is used in the post-processing step to generate the final output which is
$\cost_{R(k)}(\{0,\ldots,n\} \times \{n,\ldots,2n\})$.

\subsection{Elementary primitives}
\label{subsec:primitives}

We describe some elementary functions used within the algorithm.

\medskip
{\bf The function  $\Round$}.  $\Round(J,\epsilon)$ where $J$ is an interval and $\epsilon \leq 1$ is a power of 2, is equal to the
$\epsilon$-aligned interval $J'$ of width $\mu(J)$ obtained by shifting $J$ down (decreasing its two endpoints) at
most $\epsilon \mu(J)-1$ units.

\medskip
{\bf The function \ZoomIn}. 
Recall the definition of {\em displacement} in Section~\ref{sec:preliminaries}.
The function $\ZoomIn$ takes as input a box $I \times J$, and a subinterval $I'$ of $I$ and some additional
parameters, and outputs a set of suitably aligned intervals $J'$ of width
$\mu(I')$ so that each box $I' \times J'$ has small displacement from $I \times J$.
More precisely,
for a box $I \times J$, a subinterval $I' \subseteq I$, and $0 \leq i' \leq i \leq \log(1/\theta)$, $\ZoomIn(j,I\times J,i,I',i')$ 
is the set of all $\epssub{i'+3}$-aligned intervals $J' \subseteq J$ of width $\mu(I')$, for which the displacement of $I' \times J'$ from $I \times J$
is at most $2\epssub{i }\mu(I)$.

\begin{proposition}
\label{prop:33}
Let $I$ be an interval of width $w$ and $I' \subseteq I$ of width $w'$ a divisor of $w$. Let $i' \leq i \in \{0,\ldots,\log(1/\theta)\}$.
\begin{enumerate}
\item For $J$ of width $w$, $|\ZoomIn(j,I \times J,I',i')|$ has size at most $1+32 \epssub{i-i'}w/w'$.
\item
Let $I' \times J'$ be a box.   The number of $\epssub{i+3}$-aligned width-$w$ intervals $J$ such that
$J' \in \ZoomIn(j,I \times J,i,I',i')$ is at most 33.
\end{enumerate}
\end{proposition}
\begin{proof}
Set $\Delta=\min(I')-\min(I)$. If $J' \in \ZoomIn(j,I \times J,i,I',i')$ then $|\min(J')-\Delta-\min(J)| \leq2  \epssub{i} w$.

Proof of (1).   Holding $J$ fixed, we have $\min(J') \in [\min(J)+\Delta-2\epssub{i}w,\min(J)+\Delta+2\epssub{i}w]$. This is an interval of width $4\epssub{i}w$, and the number of $\epssub{i'+3}$-aligned intervals of width $w'$
that start in this interval is at most $1+32 \epssub{i-i'}w/w'$.

Proof of (2). Holding $J'$ fixed, we have $\min(J) \in [\min(J')-\Delta-2\epssub{i}w,\min(J')-\Delta+2\epssub{i}w]$.  This is an interval of width $4\epssub{i}w$, and the number of $\epssub{i+3}$-aligned intervals of width $w$
that start in this interval is at most 33.
\end{proof}

Calling $\ZoomIn(j,I\times \cJ,i,I',i')$ with a stack $I\times \cJ$ returns the union of results $\bigcup_{J\in \cJ} \ZoomIn(j,I\times J,i,I',i')$.

\medskip
{\bf The function \INDUCED}. This is a function that takes as input a set of weighted square boxes $\cQ$ and outputs a collection 
of weighted boxes {\em induced by $\cQ$}.   For an interval $J$, and $t \le \mu(J)/2$
let $J/[t]$ denote the interval $[\min(J)+t,\max(J)-t]$.   For each $(I \times J,\kappa)$ in $\cQ$,
$\INDUCED(\cQ)$ includes $(I \times J, \kappa)$ together with boxes of the form $(I \times J/[2^i],\kappa+\frac{2^{i+1}}{\mu(I)})$
for $i \in \{0,\ldots,\log(\mu(J))-1\}$.  

\begin{proposition}
\label{prop:INDUCED}
If all boxes of $\cQ$ are certified boxes then so are all boxes of $\INDUCED(\cQ)$. 
\end{proposition}    

\begin{proof}
Note that $|J \Delta(J/[2^i])|=2^{i+1}$ and apply Proposition~\ref{prop:sym diff}.
\end{proof}

\medskip
{\bf The function  $\APM$ (Approximate pattern match)}. 
Recall from Section~\ref{sec:preliminaries} that $\cost_{\cR}(I \times J)$ is the length of the min-cost traversal of $I \times J$ in the shortcut
graph $\widetilde{G}(\cR)$. $\APM$  takes as input a stack $I \times \cJ$, $\kappa>0$ and  a set $\cR$ of certified boxes,
and outputs a subset $\cal{S}$ of $\cal{J}$ that satisfies: 
\begin{description}
\item[Completeness of \APM.] For all $J \in \cJ$ satisfying $\cost_{\cR}(I \times J) \leq \kappa \mu(I)$,  $J \in \cS$
\item[Soundness of \APM.] For all $J \in \cJ$ satisfying $\cost(I \times J) >2\kappa \mu(I)$, $J \not\in \cS$.
\end{description}
The running time is $\widetilde{O}(\mu(I)+|\cJ|+|\cR|)$.  (Notice, the subtle distinction between $\cost_{\cR}$ and $\cost$ in Soundness and Completeness.)  The implementation, described in Section~\ref{sec:apm}, is a customized variant of dynamic programming  that closely follows~\cite{CDGKS-arxiv,CDK-arxiv}.

\subsection{The mechanics of the algorithm}
\label{subsec:mechanics}

We are now ready to present  the pseudocode for \MAIN{} and the three main subroutines:
\Preprocess{} and \ProcessDense{}, and \Enumerate{}.  

\medskip
\noindent
{\bf The algorithm \MAIN}.  This algorithm inputs an integer $n$ which is a power of $2$,
$\theta \in (0,1]$ a power of 1/2, and 
 two strings $x,y$ of length $n$, and returns \ACCEPT{} 
or \REJECT.  (Recall that the error parameter $\delta$ is fixed to 1/2.)
The algorithm consists of iterations indexed by  $j \in \{1,\ldots,k\}$. $\Preprocess(j)$ creates the sets $\Bbelow(j,I,i)$ and $\SparseSample(j,I,i)$ that enable the
 level $j$ queries $\Enumerate(j,\cdot)$.
$\ProcessDense(j)$ creates sets $\cR(j)$ and $\Sparse(j,i)$ needed for $\Preprocess(j+1)$.

\begin{algorithm}[H]
\begin{algorithmic}[1]
\caption{\MAIN$(n,\theta;x,y)$}
\label{alg-main}
   
\REQUIRE $n$ is a power of 2. $|x|=|y|=n$. $\theta \in (0,1]$ is a power of 1/2.

\ENSURE If $\ed(x,y)\geq Q\theta$ then return \REJECT.  If $\ed(x,y) \leq \theta$ then return \ACCEPT{} with probability at least 1/2.

\vspace{1mm}
\hrule\vspace{1mm}

      \STATE $z\longleftarrow xy$.

      \FOR{$j \in  \{1,\dots,k\}$}

         \STATE \Preprocess$(j)$.

         \STATE \ProcessDense$(j)$.
               
      \ENDFOR

   \STATE Return \ACCEPT{} if $\APM(\{0,\ldots,n\} \times\{\{n,\ldots,2n\}\},\theta 2^{\qual_k+5}, \cR(k))$ is non-empty, otherwise return \REJECT.
   
\end{algorithmic}
\end{algorithm}

\medskip
\noindent
{\bf The subroutine \Preprocess}. On input $j$, the sets
$\Sparse(j-1,i)$ and $\cR(j-1)$ created by $\ProcessDense(j-1)$ are used to produce the
sets $\Bbelow(j,I,i)$ and $\SparseSample(j,I,i)$ for $I \in \Intervals(w_j)$ and $i \in \{0,\ldots,\log(1/\theta)\}$.
To begin, the set of weighted $w_{j-1}$-boxes $\cR(j-1)$  is partitioned into sets
$\cR(j-1,I)$, with $I' \times J'$ assigned to $\cR(j-1,I)$ for $I' \subseteq I$. 
For each $i$ and $I$:

\begin{enumerate}
\item The set $\Sparse(j-1,i) \subseteq \Intervals(w_{j-1})$ was produced by $\ProcessDense(j-1)$.
$\SparseSample(j,I,i)=\emptyset$  if $\Sparse(j-1,i)$ contains no subintervals of $I$, and otherwise
is an independent random sample (multiset) of size $\log(n)^{\theta(1)}$ selected from the subsets of $I$ belonging to $\Sparse(j-1,i)$.
\item  Run $\APM$ with input stack $I \times \Intervals(w_j,\epssub{i+3})$ and $\cR(j-1,I)$ to determine
the set of intervals $J \in \Intervals(w_j,\epssub{i+3})$ that are suitably close to $I$ 
in the shortcut graph $\widetilde{G}(\cR(j-1,I))$.
\end{enumerate}

\begin{algorithm}[H]
\begin{algorithmic}[1]
\caption{\Preprocess$(j)$}
\label{alg-preprocess}
   
\REQUIRE $j\in[k]$. Levels $j=1,\dots,j-1$ were already processed. Uses
sets $\Sparse(j-1,i)$ and $\cR(j-1)$ constructed by $\ProcessDense(j-1)$, where $\cR(0)=\emptyset$.

\ENSURE $\Bbelow(j,I,i)$, and 
$\SparseSample(j,I,i)$ for  $(I,i) \in \Intervals(w_j)  \times \{0,\ldots,\log(1/\theta)\}$.

\vspace{1mm}
\hrule\vspace{1mm}
\STATE Partition $\cR(j-1)$ into $\{\cR(j-1,I):I \in \Intervals(w_j)\}$ where $I' \times J'$
is in $\cR(j-1,I)$ if $I' \subseteq I$.
   \FOR{$i \in \{0, \ldots, \log 1/\theta\}$}
   \FOR{$I \in \Intervals(w_j)$}
      \IF {$j=1$} \STATE $\SparseSample(j,I,i)\longleftarrow \emptyset;\; \Bbelow(j,I,i) \longleftarrow \emptyset$
      \ELSE

       \IF {$\Intervals(w_{j-1};I) \cap \Sparse(j-1,i) = \emptyset$} \STATE $\SparseSample(j,I,i) \longleftarrow \emptyset$
       \ELSE \STATE  Make $30 \log n$ independent uniform selections from $\Intervals(w_{j-1};I) \cap \Sparse(j-1,i)$ to obtain $\SparseSample(j,1,i)$
       \ENDIF
       \STATE $\Bbelow(j,I,i) \longleftarrow \APM(I \times \Intervals(w_{j},\epssub{i+3}),\epssub{i-\qual_{j-1}-5}, \cR(j-1,I))$
     \ENDIF 
   \ENDFOR
   \ENDFOR
\end{algorithmic}
\end{algorithm}

\medskip
{\bf The subroutine \Enumerate}.  The creation of $\Bbelow(j,I,i)$ and $\SparseSample(j,I,i)$ by $\Preprocess$
enables the query $\Enumerate(j,\cdot)$, which implicitly classifies all level $j$
candidates $\candidate{i}{I \times J}$ as \CLOSE{} or \FAR{} subject to: 
\begin{description}
\item[Completeness of $\Enumerate$.] If $\ncost(I \times J) \leq \epssub{i}$ then with high probability
$\candidate{i}{I \times J}$ is \CLOSE.
\item[Soundness of $\Enumerate$.] If $\ncost(I \times J) >  \epssub{i-\qual_{j-1}-6}$ then $\candidate{i}{I \times J}$ is \FAR.
\end{description}
$\Enumerate(j,\cdot)$ takes
a stack $I \times \cJ$ and $i \in \{0,\ldots,\log(1/\theta)\}$ with $I \in \Intervals(w_j)$ and $\cJ \subseteq \Intervals(w_j,\epssub{i+3})$ and returns  $\{J \in \cJ:\candidate{i}{I \times J}$ \text{ is } \CLOSE\}.
$\cS$ accumulates the set of intervals to be output.  For $j=1$, $\ED(z_I,z_J,\epsilon)$ is run for
each $J \in \cJ$ and $\cS$ is the set of accepted $J$. For $j>1$, $\cS$ is the union of two sets.
The first is
$\Bbelow(j,I,i) \cap \cJ$ found by $\Preprocess(j)$.   The second  is obtained by identifying (as described below) a 
small subset $\cK \subseteq	 \cJ$, testing each $J \in \cK$  using $\ED$, and adding $J$ to $\cS$ if $z_J$ is
suitably close to $z_I$.    To identify $\cK$, for each  $(I',i') \in \SparseSample(j,I,i) \times \{0,\ldots,i\} $ 
use
$\ZoomIn$ to identify the set $\cJ'$ of $J' \in \Intervals(w_{j-1},\epssub{i'+3})$ such that  $I' \times J'$ has displacement
at most $2\epssub{i}\mu(I)$ from  $I \times J$.
Recursively use $\Enumerate(j-1,I' \times \cJ',i')$ to select $\cS'=\{J' \in \cJ': 
\candidate{i'}{I' \times J'} \text { is }\CLOSE\}$.   $\cK$ consists of those $J$ for which $I \times J$ has small displacement
from $I' \times J'$ for some $J' \in \cS'$.

The loops on $i',I'$ (line 11-21) produce $\cK \subseteq \cJ$. For each $J \in \cK$, 
$\ED$ is run on $z_I,z_J$.  The loop on $I'$ is over $\SparseSample(j,I,i')$.  The subset $\cK$ of $\cJ$ depends
on the random sample $\SparseSample(j,I,i')$  of $\Sparse(j-1,i') \cap \Intervals(w_{j-1};I)$.
The following definitions highlight this dependence.

\begin{itemize}
\item
For $\candidate{i}{I \times J}$, let $I' \in \Intervals(w_{j-1};I)$ and $i' \in \{0,\ldots,i\}$. 
The pair $(I',i')$ is a {\em marker%
\footnote{We call it {\em marker} as in genomics, where a short DNA sequence identifies a gene. Similarly here, a marker for $z_I$
is its substring $z_{I'}$ which is relatively rare in $z$, i.e., $I'$ belongs to $\Sparse(j-1,i')$.}
for the candidate $\candidate{i}{I \times J}$} 
if $I' \in \Sparse(j-1,i')$ and there is some $J' \in \ZoomIn(j,I\times J, i, I',i')$ such that  
$\candidate{i'}{I' \times J'}$ is classified as \CLOSE{}. 
When lines (13-18) are executed for a marker $(I',i')$, $J$ is added to $\cK$ in line 17. Ideally, $\cK$ will consist of all intervals $J$ identifiable by their markers.
\item $\winners(j,I \times  J,i,i')=\{I' \in \Sparse(j-1,i') \cap \Intervals(w_{j-1};I): (I',i')$ is a marker
for $\candidate{i}{I \times J}\}$.
We will be interested in situations when for some $i'\le i$ there will be many markers, namely, $|\winners(j,I \times  J,i,i')| \ge \frac{1}{3}|\Sparse(j-1,i') \cap \Intervals(w_{j-1};I)|$, so that with high probability $\SparseSample(j-1,I,i')$ will contain a marker that will identify $J$.
\end{itemize}

%

\begin{algorithm}[H]
\begin{algorithmic}[1]
\caption{$\Enumerate(j,I\times \cJ,i)$}
\label{alg-enumerate}
   
\REQUIRE $j\in[k]$, $I \in \Intervals(w_j)$, $\cJ \subseteq \Intervals(w_j,\epssub{i+3})$, $i \in \{0,\ldots,\log(1/\theta)\}$, levels $1,\dots,j-1$ were already processed, and level $j$ was preprocessed.

\ENSURE Returns  $\cS=\{J \subseteq \cJ: \candidate{i}{I\times J} \text{ is }\CLOSE\}$

\vspace{1mm}
\hrule\vspace{1mm}

   \IF{$j=1$}
      \STATE Initialization: $\cS\longleftarrow \emptyset$.
      \FOR{$J \in \cJ$}
         \IF{$\ED(z_I,z_J,\epssub{i})$ returns \ACCEPT}
            \STATE Add $J$ to $\cS$. \cmt{$\candidate{i}{I \times J}$ {\em  is classified as \CLOSE}}
         \ENDIF            
      \ENDFOR
   \ELSE  
      \STATE \cmt{$j>1$}
 
      \STATE Initialization: $\cS\longleftarrow \Bbelow(j,I,i) \cap \cJ$, $\cJ \longleftarrow \cJ \setminus \cS$, $\cK \longleftarrow \emptyset$.
      \FOR{$i' \in \{0, \ldots, i\}$}
      \FOR{$I'\in \SparseSample(j,I,i')$}
      
         \STATE $\cJ' \longleftarrow \ZoomIn(j,I\times \cJ, i, I',i')$.
         \STATE $\cS' \longleftarrow \Enumerate( j-1, I' \times \cJ',i')$.

         \FOR{$J \in \cJ$}

            \IF{$\ZoomIn(j,I\times J,i,I',i') \cap \cS' \neq \emptyset$} 
           \STATE Add $J$ to $\cK$ 
           \ENDIF
           \ENDFOR
           \ENDFOR
          \ENDFOR
           \FOR{$J \in \cK$}
              
             \IF{ $\ED(z_I,z_J,\epssub{i})$ returns \ACCEPT}
               \STATE Add $J$ to $\cS$.  \cmt{$\candidate{i}{I \times J}$ {\em  is classified as \CLOSE}}
            \ENDIF
            
         \ENDFOR
    
   \ENDIF   

   \STATE Return $\cS$.
\end{algorithmic}
\end{algorithm}

\medskip

\begin{algorithm}[H]
\begin{algorithmic}[1]
\caption{ProcessDense$(j)$}
\label{alg-process}
   
\REQUIRE $j\in[k]$. Levels $1,\dots,j-1$ were already processed, and level $j$ was preprocessed.

\ENSURE For each  $i \in \{0,\ldots,\log(1/\theta)\}$, specify the set $\Sparse(j,i) \subseteq \Intervals(w_j)$ and
specify the sets $\Bdense(j,I,i)$  for all intervals $I \in \Intervals(w_j) \setminus \Sparse(j,i)$.  


\vspace{1mm}
\hrule\vspace{1mm}

   \FOR{$i=0,\ldots,\log 1/\theta$}

   \STATE Initialization: $\cT=\Intervals(w_j)$. 
   \STATE Initialization: $\Sparse(j,i)=\emptyset$.

\IF {$i\le \qual_j$}
\STATE Set $\Bdense(j,I,i)=\Intervals(w_j,\epssub{i+3})$  for every $I \in \cT$.
\ELSE 
\STATE \cmt{$i>\qual_j$}
   \WHILE{$\cT$ is non-empty}

      \STATE Pick $I\in \cT$.
      
      \STATE Let $\cS$ be the subset of $\Intervals(w_j,\epssub{i+3})$ obtained by including each element independently
with probability $p:=\min(1,(c_0 \log n)/d_j)$. 
      
      \IF{$\Enumerate(j,I\times \cS,i)$ has less than $p\cdot d_j$ elements }

         \STATE  Add $I$ to $\Sparse(j,i)$, and $\cT=\cT \setminus \{I\}$. \cmt{$I$ is declared sparse.}
         
      \ELSE
 
         \STATE  \cmt{$I$ is declared dense and used as a pivot.}

\STATE $h_1 \longleftarrow i-\qual_{j-1} - 7$; $h_2\longleftarrow i -2\qual_{j-1}-14$.  \cmt{Since $i>\qual_j$, $h_1,h_2 > 0$.}
         \
         \STATE $\cX \longleftarrow \Enumerate(j,I \times \Intervals(w_j), h_1)$.

         \STATE $\cY' \longleftarrow \Enumerate(j,I \times \Intervals(w_j,\epssub{h_2+3}),h_2)$.
         
         \STATE $\cY \longleftarrow \{ J \in \Intervals(w_j,\epssub{i+3}) : \Round(J, \epssub{h_2+3}) \in \cY' \}$.

         \STATE $\Bdense(j,I',i)\longleftarrow \cY$ for each $I'\in \cX$.

         \STATE  $\cT \longleftarrow \cT \setminus \cX$.
         
      \ENDIF
     
   \ENDWHILE
 \ENDIF
\ENDFOR

\STATE \cmt{Convert stack of \CLOSE{} boxes into weighted boxes.}
\STATE  $\cQ(j) \longleftarrow \{(I\times J,\epssub{i-\qual_j}): i \in \{0,\ldots,\log(1/\theta)\}, I \in \Intervals(w_j), J \in \Bdense(j,I,i)\}$.

\STATE $\cR(j)\longleftarrow\INDUCED(\cQ(j))$.

   
\end{algorithmic}
\end{algorithm}

\medskip
{\bf The procedure \ProcessDense}. This takes as input a level number $j$.  The procedure
corresponds closely to the procedure Dense Strip Removal in~\cite{CDGKS-arxiv}.

For each
$i \in \{0,\ldots,\log(1/\theta)\}$ the procedure builds a set $\Sparse(j,i) \subseteq \Intervals(w_j)$
and also builds sets $\Bdense(j,I,i) \subseteq \Intervals(w_j,\epssub{i+3})$ for every $I \in \Intervals(w_j)\setminus \Sparse(j,i)$.
This is done by processing the intervals of $\Intervals(w_j)$; when interval $I$ is processed it is either assigned to $\Sparse(j,i)$
or the set $\Bdense(j,I,i)$ is constructed.  We keep track of a subset $\cT \subseteq \Intervals(w_j)$ of {\em unprocessed} intervals. This set
is initialized to $\Intervals(w_j)$ and the iteration ends when $\cT=\emptyset$.  We proceed in rounds.  In a round we select an arbitrary $I$ 
from $\cT$.  We perform a test (see ''Testing potential pivots in $\ProcessDense$'' in Section~\ref{subsec:randomness}) to decide whether to put it in $\Sparse(j,i)$.  If $I$ is not placed in $\Sparse(j,i)$ 
then $I$ is designated the {\em pivot} for that round.  We then call $\Enumerate$ on the stack $I \times \cT$ (with suitable parameters) 
to determine the subset $\cX$ of $\Intervals(w_j)$, we call $\Enumerate$ on the stack $I \times \Intervals(w_j,\kappa)$
(for a suitable $\kappa \geq \epssub{i}$) to determine  $\cY'\subseteq \Intervals(w_j,\kappa)$ and we let $\cY$ be the set of intervals from $\Intervals(w_j,\epssub{i+3})$ which round to an interval in $\cY'$.
 We then define $\Bdense(j,I',i)=\cY$ for all $I' \in \cX$,
and remove $\cX$ from $\cT$, to complete the round.

The parameters used in the above calls are expressed in terms of $h_1$ and $h_2$ introduced in the pseudocode.   
The particular choice $h_1$ and $h_2$ is motivated by both the correctness analysis and the time analysis (Section~\ref{subsec:time analysis}).

In the sequel, we will need the following definition and observation.

\medskip
{\bf Approved Candidate.}  A candidate $\candidate{i}{I \times J}$ is said to be {\em approved} if $I \not\in \Sparse(j,i)$ and
$J \in \Bdense(j,I,i)$.  Note that the boxes in $\cQ(j)$ are in one-to-one correspondence with the
approved candidates, with $(I \times J, \epssub{i-\qual_j}) \in \cQ(j)$ if and only if $\candidate{i}{I \times J}$ is approved.
All candidates of the form $\candidate{i}{I \times J}$ are approved for $i\le \qual_j$.

\begin{proposition}
\label{prop:sparse empty}
At level $k$, the sets $\Sparse(k,i)$ are empty for all $i \in \{0,\dots,\log(1/\theta)\}$.
\end{proposition}

\begin{proof}
Since $d_k=1$, the set $\cS$ created in line (10) is all of $\Intervals(w_j,\epssub{i+3})$ which, in particular includes $I$. 
The set returned by $\Enumerate$ in line (11) includes $I$ and so the if condition fails, and $I$ is not added to $\Sparse(k,i)$.
\end{proof}

%
%
%
%
%

\subsection{The use of randomization}
\label{subsec:randomness}

Randomization is used in three parts of the algorithm: the subroutine $\ED$, the construction of $\SparseSample$ during $\Preprocess$
and in $\ProcessDense$, each time we test a selected $I \in \cT$ to decide whether it is a pivot.  We discuss each of these uses below.

\medskip
{\bf The subroutine $\ED$.} 
$\ED$ takes calling parameters $(n',\theta',\delta';x',y')$.  By our
assumption $\delta'$ is fixed to $n^{-12}$ for all calls.
The gap-soundness and completeness conditions for $\ED$ guarantee that
if $\ed(x',y') >Q'\theta' n'$ then $\ED$ returns \REJECT{},
and if $\ed(x',y') \leq \theta' n'$ then $\ED$ returns \ACCEPT{} with
probability at least $1-n^{-12}$.   Say that an execution of $\ED(n',\theta',n^{-12};x',y')$ fails if  $\ed(x',y') \leq \theta' n'$ and $\ED$ returns \REJECT{}.  
We will introduce an event $\SG$ that no call to $\ED$ fails. 

To simplify the analysis, we make the following assumption: when
we run $\MAIN$ we pregenerate a single string $\BSG$ of $b$ random bits where
$b$ is an upper bound on the number of random bits used in any call to $\ED$.  
In every call to $\ED$ we use (a prefix of) $\BSG$ to provide the random bits
for the call.  This makes all
calls to $\ED$ deterministic, and also ensures that if the algorithm
makes multiple calls
to $\ED$ with the same input parameters then all such calls yield the same output.

Reusing random bits for different calls of $\ED$ makes these calls
dependent, but this is irrelevant to the analysis.
The proof of correctness relies only on the fact that the event $\SG$ holds.

We now upper bound the probability that there is a call that does not
succeed.
Every possible input tuple
$(n',\theta',n^{-12};x',y')$ for $\ED$ satisfies that
$n'$ is a power of 2 with $n' <n$, $x',y'$ are substrings of $z=xy$ of
length $n'$, and $\theta'$ is an integral power of 1/2.    We may assume
that $\theta' \geq 1/n$ since for $\theta'<1/n$ we may assume that
$\ED$ is the deterministic algorithm that returns \ACCEPT{} if $x=y$
and \REJECT{} otherwise.    Let $\SG$ denote the event that for all possible choices of input parameters $(n',x',y',\theta')$ 
with $\theta' \geq 1/n$, the choice of random bits succeeds.

The number of possible choices
of input parameters  for which randomness is used is at most $4n^2\log^2(n)$.
(There are at most $\log(n)$ ways to choose $n'$, and to choose $\theta'$,
and at most $2n$ ways to choose the starting location of  $x'$ and of $y'$.)
Thus by a union bound, the probability that $\SG$ does not hold
is at most $n^{-8}$.

\medskip
{\bf The construction of $\SparseSample$}. $\SparseSample(j,I,i)$ is a random sample of $\Sparse(j-1,i)$ generated during \Preprocess$(j)$.
What we want from this sample is that for each $i' \in \{0,\ldots,i\}$, if a nontrivial fraction of $\Sparse(j-1,i)$ belongs
to the set of markers $\winners(j,I \times J,i,i')$ then $\SparseSample(j,I,i)$ should include a member of $\winners(j,I \times J,i,i')$.
(Note: for the purposes of this discussion, the exact technical definition of $\winners(j,I \times J,i,i')$ is unimportant,  we only need that for each $j,I,J,i,i'$,  $\winners(j,I \times J,i,i')$ and $\Sparse(j-1,i)$
are completely determined after iteration $j-1$,  and
$\winners(j,I \times J,i,i') \subseteq \Sparse(j-1,i)$.) 
Formally, we say that $\SparseSample(j,I,i)$ {\em fails} 
for $J \in \Intervals(w_j,\epssub{i+3})$
and $i' \in \{0,\ldots,i\}$ if  $|\Sparse(j-1,i) \cap \Intervals(w_{j-1};I)|>0$, $|\winners(j,I \times J,i,i')| \geq |\Sparse(j-1,i)\cap \Intervals(w_{j-1};I)|/3$ 
and $\SparseSample(j,I,i) \cap \winners(j,I \times J,i,i') = \emptyset$.  
Since $\winners(j,I \times J,i,i')$ is completely determined by the
end of iteration $j-1$, and $\SparseSample(j,I,i)$ is an independent sample of $30 \log n$
elements from $\Sparse(j,I,i)$ selected during iterations $j$,  
the probability that $\SparseSample(j,I,i)$ fails for $J,i'$ is at most
$(1-1/3)^{30 \log n} \leq n^{-10}$.  There are at most $n$ pairs $J,i'$ so  the probabibility that $\SparseSample(j,I,i)$
fails for some $J,i'$ is at most $n^{-9}$. There are at most $n$ triples $j,I,i$ so the probability
that some $\SparseSample(j,I,i)$ fails is at most $n^{-8}$. 
We denote by $\BSS^j$ the random bits that are used at iteration $j$ to generate the samples from $\Sparse(j,I,i)$ for all $I$ and $i$.

\medskip
{\bf Testing potential pivots in $\ProcessDense$}.  During the while loop for $I \in \cT$ of $\ProcessDense$, we make a random
selection of a set $\cS$, and this choice affects whether $I$ is assigned to $\Sparse(j,i)$ or becomes a pivot. The constant $c_0$ in line (10) is chosen below to
satisfy certain technical conditions.
We denote by $\BPD^j$ the random bits used at iteration $j$ to generate sets $\cS$ where we make the simplifying assumption 
that there is a designated block of bits for each possible $I\in \Intervals(w_j)$ and $i$ to select the corresponding $\cS$. (Some of the blocks
might be unused.)

There are two bad events that depend on the choice of $\cS$:
\begin{enumerate}
\item $|\Enumerate(j,I \times \Intervals(w_j,\epssub{i+3}),i)|<d_j/2$ and $I$ is not assigned to $\Sparse(j,i)$.
\item $|\Enumerate(j,I \times \Intervals(w_j,\epssub{i+3}),i)|>2d_j$ and $I$ is assigned to $\Sparse(j,i)$.
\end{enumerate}

For both of the bad events, we observe that (i) for any input $(j,I\times \cJ,i)$,
$\Enumerate(j,I \times \cJ,i)$ returns the stack of candidates
$\candidate{i}{I \times J}$ that are classified as $\CLOSE$ among $I\times \cJ$,  and
(ii) the classification of
level $j$ candidates as $\CLOSE$ or $\FAR$ is completely deterministic given the random
bits $\BSG$ for $\ED$, and the random bits $\BSS^{\le j}$ and $\BPD^{\le j-1}$ for the first $j-1$ iterations and $\Preprocess(j)$. 
Thus, for the random sample $\cS$ of $\Intervals(w_j,\epssub{i+3}),i)$, where
each interval is placed in $\cS$ independently with probability $p$, 
$\frac{1}{p}|\Enumerate(j,I \times \cS,i)|$ is an estimate of  
$|\Enumerate(j,I \times \Intervals(w_j,\epssub{i+3}),i)|$, and the 
bad events can only occur if this estimate is sufficiently inaccurate.
For suitably large $c_0$, a simple Chernoff-Hoeffding bound shows that for each $(I,i)$ the probability of a bad event
is at most
 $n^{-10}$,
and summing over the at most $O(n)$ such pairs, the probability of a bad event  is at most $n^{-9}$. 
We say $\ProcessDense$ has successful sampling if no such bad event occurs.

\medskip
{\bf Successful randomization}.
An execution of $\MAIN$ has {\em successful  randomization} if all calls to $\ED$ are correct, all calls to $\SparseSample$ are successful,
and  $\ProcessDense$ has successful sampling. We denote the event
of successful randomization by $\SR$.  By the above, $\Pr[\SR] \geq 
1-1/n^7$.

%

\subsection{The properties enforced by $\MAIN$.}
\label{subsec:properties}

In this section we state and prove a theorem that states the main properties enforced by $\MAIN$.
By hypothesis, $\ED$  is a gap algorithm for edit distance satisfying $\gapcondition(T',\zeta',Q')$. We want to show that$\MAIN$
satisfies $\gapcondition(T,\zeta,Q)$ with $T=T'+1/6$ and suitably chosen
 $\zeta>0$ and $Q\geq 1$  (depending only on $T'$,$\zeta'$ and $Q'$).
As in the discussion in 
Section~\ref{subsec:objects}, we say that the level $j$ candidate 
$\candidate{i}{I \times J}$ is classified as \CLOSE{} if $\Enumerate(j,I \times \{J\},i)$ returns  $\{J\}$ and is classified as \FAR{} if $\Enumerate(j,I \times \{J\},i)$ returns 
$\emptyset$.

\begin{theorem}
\label{thm:inductive properties}
Assume that $\ED$ is a gap algorithm for edit distance satisfying $\gapcondition(T',\zeta',Q')$.
Consider a run of $\MAIN$ on input $(n,\theta,1/2;x,y)$ where $n^{-\zeta'} \leq \theta \leq 1$ and $|x|=|y|=n$, that
meets the conditions for successful randomization.

For all $j \in \{1,\ldots,k\}$, $i \in \{0,\ldots,\log(1/\theta)\}$, $I \in \Intervals(w_j)$, $J \in \Intervals(w_j,\epssub{i+3})$, $\cJ \subseteq \Intervals(w_j,\epssub{i+3})$:
\begin{description}
\item[Soundness of $\Bbelow$.]  If $J \in \Bbelow(j,I,i)$ then $\ncost(I \times J) \leq \epssub{i-\qual_{j-1}-6}$.
\item[Completeness of $\Bbelow$.] If $\ncost(I \times J) \leq \epssub{i}$ then (i) $J \in \Bbelow(j,I,i)$  or (ii)  there exists an $i' \leq i$ such that  $|\winners(j,I \times J,i,i')| > \frac{1}{3}|\Intervals(w_{j-1};I) \cap \Sparse(j-1,i')|$.  
\item[Consistency of $\Enumerate$.] 
$J \in \Enumerate(j,I \times \cJ,i)$ if and only if $J \in \Enumerate(j,I \times \{ J \},i)$.
If $J \in \Enumerate(j,I \times \{ J \},i)$ then  $\candidate{i}{I \times J}$ is classified as \CLOSE.
\item[Soundness of $\Enumerate$.] If  $\candidate{i}{I \times J}$ is classified as \CLOSE{} then $\ncost(I \times J) \leq \epssub{i-\qual_{j-1}-6}$.
\item[Completeness of $\Enumerate$.]If $\ncost(I \times J) \leq \epssub{i}$
then $\candidate{i}{I \times J}$ is classified as \CLOSE.
\item[Validity of $\Sparse$.] $I \in \Sparse(j,i)$ implies that $\Enumerate(j,I \times \Intervals(w_j,\epssub{i+3}))$ has size at most $2d_j$.
\item[Soundness of $\Bdense$.] If $J \in \Bdense(j,I,i)$ then $\ncost(I \times J) \leq \epssub{i-\qual_j}$.
\item[Completeness of $\Bdense$.] If  $I \not\in \Sparse(j,i)$ and $\ncost(I \times J) \leq \epssub{i}$ then $J \in \Bdense(j,I,i)$.
\item[Soundness of $\cR(j)$.]  Every box in $\cR(j)$ is correctly certified, i.e., $(I \times J,\kappa) \in \cR(j)$ implies $\ncost(I \times J) \leq \kappa$.
\item[Completeness of $\cQ(j)$.] If $I \not\in \Sparse(j,i)$ and $\ncost(I \times J) \leq \epssub{i}$ then $(I \times J,\min(1,\epssub{i-\qual_j})) \in \cQ(j)$ 
\end{description}
\end{theorem}

The proof of this theorem is by induction on $j$.  For fixed $j$ when we prove a property we assume the properties
listed above it hold.   With the exception of the Completeness of $\Bbelow$, which we defer to the next subsection, the proofs  are straightforward. 


\medskip
{\bf Proof of Soundness of $\Bbelow$.}
For $j=1$, the requirement is vacuously satisfied.  
Suppose $j > 1$. By Soundness of $\cR(j-1)$,
every box in $\cR(j-1,I)$ is certified. If $J \in \Bbelow(j,I,i)$, then the pseudocode
implies that $J \in \APM(j, I \times \Intervals(w_{j},\epssub{i+3},\epssub{i-\qual_{j-1}-5}, \cR(j-1,I)) $.    By definition of the soundness of \APM, $J$ is included in the output to the call of $\APM$  implies that
$\cost(I \times J) \leq 2\epssub{i-\qual_{j-1}-5}=\epssub{i-\qual_{j-1}-6}$.

\medskip
{\bf Proof of Completeness of $\Bbelow$.}  See subsection~\ref{subsec:completeness}

\medskip
{\bf Proof of Consistency of $\Enumerate$.}
We must show that whether 
$J \in \Enumerate(j,I \times \cJ,i)$ does not depend on $\cJ \setminus \{J\}$.
In the case $j=1$, $J \in \Enumerate(i,I \times \cJ,i)$ if and only $\ED(I \times J) \leq \epssub{i}$ returns \ACCEPT{} which
does not depend on $\cJ \setminus \{J\}$.  Assume $j>1$.  From the pseudocode of $\Enumerate$, $J \in \Enumerate(j,I \times \cJ,i)$  if and only if
(i) $J \in \Bbelow(j,I,i)$ or (ii) $J \in \cK$ and $\ED(z_I,z_J,\epssub{i})$ returns \ACCEPT.
Neither condition (i) nor  $\ED(z_I,z_J,\epssub{i})$ depend on $\cJ \setminus \{J\}$.
It remains to show that whether $J \in \cK$ is also independent of $\cJ \setminus \{J\}$.  Now $J \in \cK$ if and only if
there exists $i' \in \{0,\ldots,i\}$, $I' \in \SparseSample(j,I,i)$ and  $J' \in \ZoomIn(j,I \times J,i,I',i')$ such that  $J' \in \cS'$.
The set $\ZoomIn(j,I \times J,i,I',i')$ obviously doesn't depend on $\cJ \setminus \{J\}$.  
For $J' \in \ZoomIn(j,I \times J,i,I',i')$ we must have $J' \in \cJ'$, and therefore by the consistency of $\Enumerate$
at level $j-1$, $J' \in \Enumerate(j-1,I' \times \cJ',i')$ if and only if $J' \in \Enumerate(j-1,I' \times \{J'\},i')$.

\medskip
{\bf Proof of Soundness of $\Enumerate$.} $\candidate{i}{I \times J}$ is classified as \CLOSE{}
means that $J \in \Enumerate(j,I \times \{J\},i)$.  Now for this to happen either (i) $\ED(z_I,z_J, \epssub{i})$ returns \ACCEPT,
or (ii) $J \in \Bbelow(j,I,i)$.  If (i) holds then the guarantee on $\ED$ implies $\ncost(I \times J) \leq Q'\epssub{i} \leq \epssub{i-\qual_{j-1}-6}$, since $\log(Q')=\qual_0 \leq \qual_{j-1}$ for all $j \geq 1$.
If (ii) holds then the result follows from the Soundness of $\Bbelow$.

\medskip
{\bf Proof of Completeness of $\Enumerate$.}
Suppose $\ncost(I \times J) \leq \epssub{i}$.  By the Completeness of $\Bbelow$, we have (i) $J \in \Bbelow(j,I,i)$  or (ii) $\Sparse(j-1,i) \cap \Intervals(w_{j-1};I) \neq \emptyset$ and there exists an $i^* \leq i$ so that $|\winners(j,I \times J,i,i^*)| \geq \frac{1}{3}|\Intervals(w_{j-1};I) \cap \Sparse(j-1,i^*)|$.
If (i) holds, then the definition of $\Enumerate$ immediately gives $J \in \Enumerate(j,I \times \{J\},i)$.  If (ii) holds,
then the success condition for $\SparseSample(j,I,i)$ (from Section~\ref{subsec:randomness}) implies that there  is an $I^* \in \SparseSample(j,I,i^*)$ such that $(I^*,i^*)$ is a marker for $\candidate{i}{I \times J}$.  
During the execution of $\Enumerate(j,I \times J,i)$, when $i^*$ is selected in line (11) and $I^*$ in line (12), by the definition
of marker, $J$ is added to $\cK$ in line (17).
The correctness of $\ED{}$ implies
that  $\ED(I \times J,\epssub{i})$ will \ACCEPT{} in line (23) and so $J$ will be added to $\cS$.

\medskip
{\bf Proof of Validity of $\Sparse$.}
This follows immediately from the assumption that $\ProcessDense$ has successful sampling.

\medskip
{\bf Proof of Soundness of $\Bdense$.} 
For $i\le \qual_j$ the claim is trivial so we assume $i-\qual_j >0$.
Suppose $J \in \Bdense(j,I,i)$.  $\Bdense(j,I,i)$ was defined during iteration $i$ of the main loop (1-34) of $\ProcessDense(j)$, during one of the iterations of the while loop (8-23).  Let $I^*$ be the pivot during
that iteration.   Then $I \in \Enumerate(j,I^* \times \Intervals(w_j),h_1)$ and $J' \in \Enumerate(j,I^* \times \Intervals(w_j,\epssub{h_2+3}),h_2)$, 
for $J'=\Round(J, \epssub{h_2+3})$.  By the Soundness of $\Enumerate$, $\ncost(I^{*} \times I) \leq \epssub{h_1-\qual_{j-1}-6}$ and $\ncost(I^* \times J') \leq \epssub{h_2-\qual_{j-1}-6}$.  By
the triangle inequality  and Propositon~\ref{prop:sym diff}, we have $\ncost(I \times J) \leq \epssub{h_1-\qual_{j-1}-6}+\epssub{h_2-\qual_{j-1}-6} + \epssub{h_2+3}  \leq 2\epssub{h_2-\qual_{j-1}-6}=\epssub{h_2-\qual_{j-1}-7} = \epssub{i-3\qual_{j-1}-21}=\epssub{i-\qual_j}$.

\medskip
{\bf Proof of Completeness of $\Bdense$.}
Suppose $I \not\in \Sparse(j,i)$ and  $\ncost(I \times J) \leq \epssub{i}$. Since $I \not\in \Sparse(j,i)$ during iteration $i$ of the main loop (1),
there is an iteration of the while loop (8-22) of $\ProcessDense(j)$
where $I$ was removed from $\cT$.  Let $I^*$ be the pivot for that iteration.  Since
$I$ was removed from $\cT$, $I \in \cX$ during this iteration, so  $I \in \Enumerate(j,I^* \times \Intervals(w_j),h_1)$ and by the Soundness of $\Enumerate$
$\ncost(I^* \times I) \leq \epssub{h_1-\qual_{j-1}-6}$.  Let $J'=\Round(J, \epssub{h_2+3})$. It suffices to show that $J \in \cY$ for this same iteration, which would follow
from $J' \in \Enumerate(j,I^* \times \Intervals(w_j,\epssub{h_2+3}),h_2)$.   By the Completeness of $\Enumerate$ it suffices to
show that $\ncost(I^* \times J') \leq \epssub{h_2}$.  By the triangle inequality and Propositon~\ref{prop:sym diff},
$\ncost(I^* \times J') \leq \ncost(I^* \times I) + \ncost(I \times J')
\leq \ncost(I^* \times I) + \ncost(I \times J) + \epssub{h_2+3} 
\leq \epssub{h_1-\qual_{j-1}-6}+\epssub{i} + \epssub{h_2+3} \leq \epssub{h_2}/2 + \epssub{h_2}/4 + \epssub{h_2}/8 \le \epssub{h_2}$ as required.

\medskip
{\bf Proof of Soundness of $\cR(j)$.} By Proposition~\ref{prop:INDUCED}, it suffices that every box in $\cQ(j)$ is correctly certified.   In line (26) of $\ProcessDense(j)$, $(I \times J, \epssub{i-\qual_j}) \in \cQ(j)$ only if $J \in \Bdense(j,I,i)$ 
which   is correctly certified by
the Soundness of $\Bdense$.

\medskip
{\bf Proof of Completeness of $\cQ(j)$.}
 Suppose $I \not\in \Sparse(j,i)$ and $\ncost(I \times J) \leq \epssub{i}$. By the Completeness of $\Bdense$, $J \in \Bdense(j,I,i)$
and so the definition of $\cQ(j)$ implies that $(I \times J,\epssub{i-\qual_j}) \in \cQ(j)$.

\subsection{Proof of Completeness of $\Bbelow$}
\label{subsec:completeness}
Here we finish the proof of Theorem~\ref{thm:inductive properties}, by establishing the final property, whose proof
is significantly more involved than that of the others.   The proof is based on ideas from~\cite{CDGKS-arxiv}.

Consider a candidate $\candidate{i}{I \times J}$ with $\ncost(I \times J) \leq \epssub{i}$. 
We assume condition (ii) fails  and deduce $\cost_{\cR(j-1,I)}(I \times J) \leq \epssub{i-\qual_{j-1}-5} w_j$. By the definition
of $\Preprocess$ and the Completeness of $\APM$, this immediately implies condition
$J \in \Bbelow(j,I,i)$, which is condition (i).  

Fix a minimum cost traversal $\tau$ of $I \times J$.
The proof proceeds  via the following steps.

\begin{description}
\item[Step 1.]  For each
$I' \in \Intervals(w_{j-1};I)$ we specify a candidate  
$\candidate{t(I')}{I' \times \hat{J}(I')}$, which is approved in the sense defined in the description of $\ProcessDense$
in Section~\ref{subsec:mechanics}. (The
collection of boxes $\{I' \times \hat{J}(I'):I' \in \Intervals(w_{j-1};I)\}$ should be thought of as approximatly covering $\tau$.) 
\item[Step 2.]  We upper bound $\cost_{\cR(j-1,I)}(I \times J)$ as a constant times $\sum_{I'} \epssub{t(I')} w_{j-1}$ plus $8\epssub{i}w_j$.
\item[Step 3.] We show that if (ii) fails, then $\sum_{I'}\epssub{t(I')} w_{j-1}$ can be upper bounded  by a constant multiple of $\epssub{i}w_j$
\item[Step 4.] This gives that $\cost_{cR(j-1)}(I \times J)$ is at most a constant multiple of $\epssub{i}w_j$.
\end{description}

{\bf Step 1.} Specifying $\candidate{t(I')}{I \times \hat{J}(I')}$ for each $I'$.
Consider a pair $(I',i')$ where $i' \in \{0,\ldots,i\}$ and $I' \in \Intervals(w_{j-1};I)$.
 Proposition~\ref{prop:align-box} implies there 
is a level $j-1$ candidate $\candidate{i'}{I' \times J'}$ such that
$\ncost(I' \times J') \leq 2\ncost(\tau_{I'})+\epssub{i'+3}$ and 
$\displace(I' \times J',\tau_{I'}) \leq \cost(\tau_{I'})+\epssub{i'+3}w_{j-1}$.
Select such an interval $J'$ and denote it by $J_{i'}(I')$ (keeping the dependence on $\tau$ implicit.)

For each $I'$ let us define $t(I')$ to be the largest index $h\le i$ for which the candidate $\candidate{h}{I' \times J_i(I')}$ is approved
that is $I' \not\in \Sparse(j-1,h)$ and $J_{i'}(I') \in \Bdense(j-1,I',h)$. Let $\hat{J}(I')=J_{t(I')}(I')$.  We record the important properties:

\begin{proposition}
\label{prop:t(I')}  For each $I' \in \Intervals(w_{j-1};I)$:
\begin{enumerate}
\item The box $I' \times \hat{J}(I')$ satisfies $\ncost(I' \times \hat{J}(I')) \leq 2\ncost(\tau_{I'})+\epssub{t(I')+3}$ and 
$\displace(I' \times \hat{J}(I'),\tau_{I'}) \leq \cost(\tau_{I'})+\epssub{i'+3}w_{j-1}$.  
\item The candidate $\candidate{t(I')}{I' \times \hat{J}(I')}$ is approved, and hence $(I' \times \hat{J}(I'),\epssub{t(I')-\qual_{j-1}}) \in \cQ(j-1)$.
\item For any $i' \in \{t(I')+1,\dots,i\}$ either  $I' \in \Sparse(j-1,i')$ or $I' \not\in \Sparse(j-1,i')$
and $J_{i'}(I') \not\in \Bdense(j-1,I',i')$.
\end{enumerate}
\end{proposition}

\begin{proof}
The first two properties follow immediately from the definitions of $t(I')$ and $\hat{J}(I')$.  For the third property, the maximality of $t(I')$
implies that for $i' \in \{t(I')+1,\dots,i\}$, $\candidate{i'}{I' \times J_i(I')}$ is not approved, and the result follows from the definition of approved.
\end{proof}
\medskip
{\bf Step 2.} Upper bound on $\cost_{\cR(j)}(I \times J)$.

\begin{proposition}
\label{prop:R(j-1) cost}
$$\cost_{\cR(j-1,I)}(I \times J) \leq 8\epssub{i}w_j+\sum_{I' \in \Intervals(w_{j-1};I)} \epssub{t(I')-\qual_{j-1}-1}w_{j-1}.$$
\end{proposition}

(This is closely related to Lemma 4.1 of~\cite{CDGKS-arxiv} and the proof is similar.)
\begin{proof}
We transform the path $\tau$ in $G_{z}$ to a path $\tau'$ in the shortcut graph $\widetilde{G}(\cR(j-1,I))$ (see Section~\ref{sec:preliminaries})  and control the increase in cost.  Let $I_1,\ldots,I_m$ be the intervals of 
$\Intervals(w_{j-1};I)$ in order, and for $h \in [m]$,
let $i_h=t(I_h)$ and
$J_h=\widehat{J}(I_h)$. 
Let $\delta_h$ be the smallest power of 2 such that $\delta_h w_{j-1} \ge \displace(I_h \times J_h,\tau_{I_h})$.
By Proposition~\ref{prop:t(I')},
$\delta_h \leq 2\ncost(\tau_{I_h})+2\epssub{i_h+3}$, and $(I_h \times J_h,\epssub{i_h-\qual_{j-1}}) \in \cQ(j-1)$.  
Let $L=\{h \in [m]:\delta_h < 1/2\}$.  For $h \in L$, let $J_h'=J_h/[\delta_h w_{j-1}]$ (the interval
obtained by
removing the first and last $\delta_h w_{j-1}$ indices from $J_h$).  The
certified box $(I_h \times J_h', \epssub{i_h-\qual_{j-1}}+2\delta_h)$ belongs to $\cR(j-1)$, and since $I_h \subseteq I$,
it also belongs to $\cR(j-1,I)$.  Let $e_h=e_{I_h,J_h'}$ be the shortcut edge with cost 
$(\epssub{i_h-\qual_{j-1}}+2\delta_h)w_{j-1}$.
We claim (1) there is a source-sink path $\tau'$ in $\widetilde{G}(\cR(j-1,I))$ that consists of $\{e_i:i \in L\}$, plus a collection $\{H_i:i \in [m]\setminus L\}$
where $H_i$ is a horizontal path whose projection to the $x$-axis is $I_i$, plus
a collection of (possibly empty) vertical paths $V_0,V_1,\ldots,V_m$ where the $x$-coordinate of $V_i$ for $i >0$ is $\max(I_i)$
and 0 for $V_0$,
and (2)  $\cost(\tau')$ satisfies the bound of the lemma.

For the  first claim, for $h \in [m]$,
let $p_h=(i_h,j_h)$ be the first point in $\tau_{I_h}$ and define $p_{m+1}$ to be the final point of $\tau$. We will define
$\tau'$ to pass through all of the $p_h$.  Let $J_h^*$ be the vertical projection of $\tau_{I_h}$
so that $\tau_{I_h}$ traverses $I_h \times J_h^*$.   The choice of $\delta_h$ implies that for $h \in L$, $J_h' \subseteq J_h^*$. 
Define the portion $\tau'_h$ between $p_h$ and $p_{h+1}$ as follows: if $h \in L$,
climb vertically from $p_h$ to $(i_h,\min(J'_h))$ and  traverse $e_{I_{h},J'_{h}}$ and climb vertically to $p_{h+1}$
and if $h \not\in L$ then move horizontally from  $p_h$ to $(i_{h+1},j_h)$ and then climb vertically to $p_{h+1}$.

For the second claim, we upper bound $\cost(\tau')$.  
For $h \in L$, $e_{I_{h},J_h}$ has cost at most 
$(\epssub{i_h-\qual_{j-1}}+2\delta_h)w_{j-1}$, and
for $h \not\in L$, the horizontal path that projects to $I_h$ costs $w_{j-1} \leq 2\delta_h w_{j-1}$; the total cost of
shortcut and horizontal edges is at most 
$\sum_h (\epssub{i _h-\qual_{j-1}}+ 2 \delta_h)w_{j-1}$.   
The cost of vertical edges
is $\sum_{h \in L}(w_{j-1} - \mu(J'_h)) + \sum_{h \not\in L}w_{j-1} =
\sum_{h \in L} 2\delta_h w_{j-1}+\sum_{h \not\in L}w_{j-1} \leq \sum_h 2\delta_hw_{j-1}$.

The combined cost of all edges is at most 
\begin{eqnarray*}
\sum_h (\epssub{i_h-\qual_{j-1}} + 4 \delta_h) w_{j-1} & \leq & \sum_h (\epssub{i_h-\qual_{j-1}} + 8\cost(\tau_{I_h})+8\epssub{i_h+3})w_{j-1} \\
& \le & 8\cost(\tau) + \sum_h (\epssub{i_h-\qual_{j-1}}+\epssub{i_h}) w_{j-1}\\
& \leq &8\cost(\tau) + \sum_h \epssub{i_h-\qual_{j-1}-1}w_{j-1},
\end{eqnarray*}
which implies the desired bound.
\end{proof}

\medskip
{\bf Step 3.} Implication of failure of condition (ii).
We now use the failure of (ii) to obtain an upper bound on the righthand side of Proposition~\ref{prop:R(j-1) cost}.

For $i'\le i$, let $\winners_{i'} = \winners(j,I \times J,i,i')$ and $\cS_{i'}$ represent the set $\Sparse(j-1,i') \cap \Intervals(w_{j-1};I)$.  
Let $\cI' = \Intervals(w_{j-1};I)$.

The failure of condition (ii) implies:

\begin{equation}
\label{winners fail}
|\winners_{i'}| \leq \frac{1}{2}|\cS_{i'} \setminus \winners_{i'}|
\end{equation}.

Multiplying (\ref{winners fail}) by $\epssub{i'}$ 
and summing on $i'$ yields:
 \begin{equation*}
 \sum_{i' \le i} \sum_{I'\in \winners_{i'}} \epssub{i'} \leq \frac{1}{2} \sum_{i' \le i} \sum_{I'\in \cS_{i'} \setminus \winners_{i'}} \epssub{i'} .
 \end{equation*}
Switching the sums:
 \begin{equation}
 \label{eqtn winner loser}
 \sum_{I' \in \cI'} \sum_{i':I'\in \winners_{i'}} \epssub{i'} \leq \frac{1}{2} \sum_{I' \in \cI'} \sum_{i':I'\in \cS_{i'} \setminus \winners_{i'}} \epssub{i'} .
 \end{equation}

To reduce this further, we need the
 following sufficient condition for $I' \in \winners_{i'}$.


\begin{proposition}
\label{prop:marker condition} 
Suppose the candidate $\candidate{i}{I \times J}$ satisfies $\ncost(I \times J) \leq \epssub{i}$
and $\tau$ is a min-cost traversal of $I \times J$.  Let $(I',i')$ be a pair such that $I' \in \Intervals(w_{j-1};I)$ and $i' \in \{0,\ldots,i\}$. 
\begin{enumerate}
\item If $\epssub{i'}\geq 2\ncost(\tau_{I'})+\epssub{i+3}$
then $\ncost(I' \times J_{i'}(I')) \leq \epssub{i'}$.
\item If $\epssub{i'}\geq 2\ncost(\tau_{I'})+\epssub{i+3}$ and $I'\in \Sparse(j-1,i')$ then $(I',i')$ is a marker for $\candidate{i}{I \times J}$.
\end{enumerate}
\end{proposition}

\begin{proof}
For the first part, by the choice of 
$J_{i'}(I')$, we have $\ncost(I' \times J_{i'}(I')) \leq 2\ncost(\tau_{I'})+\epssub{i+3}$ and by the hypothesis of the
Proposition, this is at most $\epssub{i'}$.  

For the second part. By Completeness of $\Enumerate(j-1,\cdot)$ and the first part, $\candidate{i'}{I' \times J_{i'}(I')}$ is classified as \CLOSE{}. 
So we just have to show that $J_{i'}(I) \in \ZoomIn(j,I \times \{J\},i,I',i')$.
It suffices that $\displace(I' \times J_{i'}(I'),I \times J) \leq 2\epssub{i}w_j$.
To bound $\displace(I' \times J_{i'}(I'),I \times J)$ it suffices to bound the vertical distance from the point 
$(\min(I'),\min(J_{i'}(I')))$ to the diagonal of $I \times J$.  Let $(p,q)$ be the initial point of
$\tau_{I'}$.  By the definition of  $J_{i'}(I'))$, the vertical distance from $(\min(I'),\min(J_{i'}(I')))$ to $(p,q)$ is 
at most $\cost(\tau_{I'})+\epssub{i'+3}w_{j-1} \leq \cost(\tau)+w_{j-1}$.  By Proposition~\ref{prop:point displacement} the 
vertical distance from $(p,q)$ to the diagonal of $I \times J$ is at most $\cost(\tau)/2$.  
So $\displace(I' \times J_{i'}(I'),I \times J) \leq \frac{3}{2} \cost(\tau) + w_{j-1}$.   
By hypothesis, $\cost(\tau) \leq \epssub{i}w_j$, and  by assumption (\ref{eqn:tech assump}) 
in Section~\ref{subsec:params},  $w_{j-1} \leq \frac {\theta}{2} w_j$, and so $\displace(I' \times J_{i'}(I')),I \times J) \leq (\frac{3}{2} \epssub{i}+\frac{1}{2} \theta)w_j \leq 2\epssub{i}w_j$, as 
required.
\end{proof}

Let $\cG(I)=\{I' \in \Intervals(w_{j-1};I):t(I')<i \;\&\; \epssub{t(I')+1} \geq 2\ncost(\tau_{I'})+\epssub{i+3}\}$.  
We claim that for each $I' \in \cG(I)$, $I'\in \Sparse(j-1,t(I')+1)$. If it were not then by Part 3 of Proposition~\ref{prop:t(I')},
$J_{i'}(I') \not\in \Bdense(j-1,I',t(I')+1)$. But by Part 1 of Proposition~\ref{prop:marker condition} this would 
contradict completeness of $\Bdense(j-1,I',t(I')+1)$. Hence, for each $I' \in \cG(I)$, $(I',t(I')+1)$ is a marker.

We will combine Proposition~\ref{prop:marker condition} with inequality (\ref{eqtn winner loser}).  
The sum on the lefthand side of (\ref{eqtn winner loser}) includes
all pairs $(I',t(I')+1)$ where  $I' \in \cG(I)$ and so is bounded below by $\sum_{I' \in \cG(I)} \epssub{t(I')+1}$.
To upper bound the righthand sum of (\ref{eqtn winner loser}), we look at the inner sum corresponding to a given $I' \in \Intervals(w_{j-1};I)$.
This is a sum of $\epssub{i'}$ over those $i'$ such that $I'$ in $\Sparse(j-1,i')$ and $I'$ not in $\winners_{i'}$.

We claim that if $i'$ contributes to this sum then
\begin{equation}
 \label{sum pty}
   \epssub{i'} < 2\ncost(\tau_{I'})+\epssub{i+3}.
 \end{equation}
To see this note that if $\epssub{i'} \ge 2\ncost(\tau_{I'})+\epssub{i+3}$ then
Part 2 of Proposition~\ref{prop:marker condition} implies that $I' \not \in \cS_{i'} \setminus \winners_{i'}$, 
so $i'$ is not included in the sum.

Now in the case that $I' \in \cG(I)$ then (\ref{sum pty}) implies that
$\epssub{i'} <  \epssub{t(I')+1}$ and so $\epssub{i'} \le  \epssub{t(I')+2}$. Summing over all such $i'$,
the geometric series is at most $\epssub{t(I')+1}$.

For $I' \not \in \cG(I)$, let $v(I')$ be the least $i'$ that contributes to the sum. 
So the sum is at most $2\epssub{v(I')}$, and by (\ref{sum pty}) this is at most 
$4\ncost(\tau_{I'})+\epssub{i+2}$.
%
%

Thus (\ref{eqtn winner loser}) implies:

\begin{equation}
\label{eqtn winner loser 2}
\sum_{I' \in \cG(I)} \epssub{t(I')+1} \leq \frac{1}{2}\cdot \left(\sum_{I' \in \cG(I)} \epssub{t(I')+1}+\sum_{I' \not\in \cG(I)} 4 \ncost(\tau_{I'})+\epssub{i+2} \right).
\end{equation}

Multiplying the inequality by $2$ and substracting $\sum_{I' \in \cG(I)}\epssub{t(I')+1}$ from both sides gives:

\begin{equation}
\label{eqtn winner loser 3}
\sum_{I' \in \cG(I)} \epssub{t(I')+1} \leq \sum_{I' \not\in \cG(I)} 4 \ncost(\tau_{I'})+\epssub{i+2}
\end{equation}

Now add $\sum_{I' \not\in \cG(I)} \epssub{t(I')+1}$ to both sides:

\begin{equation}
\label{eqtn winner loser 4}
\sum_{I' \in \Intervals(w_{j-1};I)} \epssub{t(I')+1} \leq \sum_{I' \not\in \cG(I)} \epssub{t(I')+1} + \sum_{I' \not\in \cG(I)} 4 \ncost(\tau_{I'})+\epssub{i+2}
\end{equation}

For the first sum on the right, $I' \not \in \cG(I)$ implies either $\epssub{t(I')+1}=\epssub{i+1}$ or
$\epssub{t(I')+1} < 2\ncost(\tau_{I'})+\epssub{i+3}$, so you can bound this in both cases
by $2\ncost(\tau_{I'})+\epssub{i+1}$. Thus we get:


\begin{eqnarray*}
\sum_{I' \in \Intervals(w_{j-1};I)} \epssub{t(I')}w_{j-1} & \leq  &  2 \cdot \sum_{I' \in \Intervals(w_{j-1};I)} (6\ncost(\tau_{I'})+3\epssub{i+2})w_{j-1}\\
 & \leq  &  12 \cost(\tau)+2\epssub{i}w_j\\
& \leq & 14\epssub{i}w_j.
\end{eqnarray*}

{\bf Step 4.} Combining the bounds.
Combining the previous bound with the bound of Proposition~\ref{prop:R(j-1) cost} gives:

\begin{eqnarray*}
\cost_{\cR(j-1,I)}(I \times J) & \leq & 14 \epssub{i-\qual_{j-1}-1}w_j+\epssub{i-3}w_j  \\
& \leq & \epssub{i-\qual_{j-1}-5}w_j.
\end{eqnarray*}
as required to establish the Completeness of $\Bbelow$.

\subsection{Correctness of \MAIN}
\label{subsec:main proof}

We now  complete the proof that the output of \MAIN{} gives a constant factor approximation to edit distance with high probability.
As in Theorem~\ref{thm:inductive properties} we assume
that $\ED$ is a gap algorithm for edit distance satisfying $\gapcondition(T',\zeta',Q')$.
Consider a run of $\MAIN$ on input $(n,\theta,1/2;x,y)$ where $n^{-\zeta'} \leq \theta \leq 1$ and $|x|=|y|=n$.
The conclusion of the theorem has a quality parameter $Q$ which we set to $2^{\qual_k+6}$.  We must prove that
the \MAIN{} satisfies the Soundness and Completeness properties for gap algorithms from Section~\ref{sec:intro}.

The final post-processing step is a call to $\APM(\{0,\ldots,n\} \times\{\{n,\ldots,2n\}\},\theta2^{\qual_k+5}, \cR_k)$, and
the algorithm returns \ACCEPT{} or \REJECT{} according to the output of this call.  We will apply the Soundness and
Completeness of $\Bbelow$ (with $j=k+1$) by reinterpreting this final step as asking whether  $\{n,\ldots,2n\} \in \Bbelow(k+1,\{0,\ldots,n\},\log(1/\theta)-\qual_k-6)$ (where
$w_{k+1}=n$).  The Soundness and Completeness of $\Bbelow$ extends (with no change) to this case.
 Thus if the algorithm returns \ACCEPT{}, then  $\ncost(\{0,\ldots,n\},\{n,\ldots,2n\}) \leq
\theta2^{\qual_k+6}=\theta Q$, and the gap-algorithm satisfies Soundness.  For Completeness, assume   $\ed(x,y) \leq \theta$.
The Completeness  of $\Bbelow$ extends (with no change) to this case.  We conclude that  
(i) $J \in \Bbelow(k+1,\{0,\ldots,n\},\log(1/\theta)-\qual_k-5)$  or (ii)  
there exists an $i' \leq i$ such that  $|\winners(k+1,I \times J,i,i')| > \frac{1}{3}|\Intervals(w_{k};I) \cap \Sparse(k,i)|$.    
Since $d_k=1$, Proposition~\ref{prop:sparse empty} implies
all sets $\Sparse(k,i)$ are empty, so $\winners(k+1,I \times J,i,i')$ are also empty but (ii) requires them to be non-empty. 
Hence, (ii) can not hold, and so (i) holds, which implies
 \MAIN{} must \ACCEPT, and so Completeness holds.

\subsection{Time analysis}
\label{subsec:time analysis}
In this subsection, we upper bound the expected running time of \MAIN{}
conditioned on the event $\SR$ of successful randomization, 
in terms of the algorithm parameters $w_1,\ldots,w_k$,
and $d_0,\ldots,d_k$.  These parameters will be optimized in the next subsection.

\begin{theorem}
\label{thm:time}
Suppose that $\ED$ is a gap algorithm for edit distance satisfying $\gapcondition(T',\zeta',Q')$. For $\theta \geq n^{-\zeta'}$
the expected running time of $\MAIN(n,\theta,1/2;x,y)$ conditioned on 
$\SR$ is upper-bounded by:

\begin{equation}
\label{eqn:time bound}
\widetilde{O}\left(\sum_{j=1}^k \frac{n}{\theta^2 w_jd_j}(\sum_{h=1}^j d_{h-1}w_h^{1+1/T'})\right).
\end{equation}
\end{theorem}

The above theorem is not quite sufficient for our purposes since
it gives only an expected upper bound on the running time of the algorithm,
while we want an absolute upper bound.   We can replace
the expected upper bound by an absolute upper bound by the following
routine modification of $\MAIN$.
 On a given input, use the above theorem to determine a number
$\tau^*$ which is at least six times the expected upper bound on running
time given by the above theorem.  Then the probability that $\MAIN$
takes more than $\tau^*$ steps is at most $1/6$.  So  we run $\MAIN$
but terminate with \REJECT{} if it reaches $\tau^*$ steps.  This converts
the expected running time to an absolute bound on running time, but now
the completeness error (the probability of false rejection) is increased
from 1/2 to 2/3.  But  by running this algorithm twice and accepting if either
run accepts we restore the completeness error to below 1/2.

Combining the above theorem with this modification gives an algorithm
satisfying the correctness properties proved for $\MAIN$ and having
an absolute upper bound on time given as in the above theorem.

We now proceed to the proof of Theorem~\ref{thm:time}.
\begin{proof}
Recall from Section~\ref{subsec:randomness} that
successful randomization means: (1) All calls to $\ED$ return
correct answers, (2) All calls to $\SparseSample$ are successful and (3)
$\ProcessDense$ has successful sampling.  

Recall that $\BSG$ is the sequence of random bits pregenerated for the calls to $\ED$
(as described in Section~\ref{subsec:randomness}).  For $j \in [k]$,
$\BSS^j$ are the random bits generated to select $\SparseSample$'s in $\Preprocess(j)$ at iteration $j$ of the algorithm, and 
$\BPD^j$ are the random bits generated to select sets $\cS$ in $\ProcessDense(j)$ at iteration $j$ 
(also as described in Section~\ref{subsec:randomness}).
Let $B^{\leq j}$ denote the random bits $\BSG,\BSS^1,\BPD^1, \ldots, \BSS^{j},\BPD^j$.
We introduce the following events:

\begin{itemize}[align=left]
\item[$\SG$] All calls to $\ED$ return correct answers.
\item[$\SSE(j)$] All calls to $\SparseSample$ during iteration $j$ are successful.
\item[$\SSE(\leq j)$]All calls to $\SparseSample$ through the end of iteration $j$ are successful.
\item[$\PD(j)$] $\ProcessDense$ has successful sampling during iteration $j$.
\item[$\PD(\leq j)$] $\ProcessDense$ has successful sample through the end of iteration $j$.
\item[$\SR$] Successful randomization, i.e. $\SG \wedge \SSE(\leq k) \wedge \PD
(\leq k)$.
\end{itemize}

We will argue that
the expected running time of \MAIN{} conditioned on $\SR$ is bounded by (\ref{eqn:time bound}).
In the bound, the outer sum on $j$ corresponds to iterations of \MAIN{}.  We will show that the cost of iteration $j$
is bounded by the inner sum.  When we analyze iteration $j$ we fix the randomness  $B^{\leq j}$ in such a way that $\SG \wedge \SSE(\leq j-1) \wedge \PD(\leq j-1)$ holds. The cost  of iteration $j$ is bounded conditioned on these
fixed random bits and subject to requirement $\SSE(j) \wedge \PD(j)$.

As a first step, we need a bound on the  running time for $\Enumerate$.  
Recall that fixing the random bits $B^{\leq j-1}$ and $\BSS^j$ makes $\Enumerate(j,\cdot)$ run deterministically.
In the lemma below, we condition on $(B^{\leq j-1},\BSS^j)=\beta^{*j}$ and
consider the expected time of $\Enumerate(j,I \times \cJ,i)$ where $\cJ$ is a set
of intervals chosen
according to any distribution (possibly depending on $\beta^{*j}$) in which no set appears in $\cJ$ with probability more than some fixed bound $p$. 

\begin{lemma}
\label{lem:enumerate time}
Let $p\in[0,1]$, $j\in[k]$, $I \in \Intervals(w_j)$, and $i \in \{0,\ldots,\log(1/\theta)\}$.   Let $\beta^{*j}$ be an assignment of the random bits
$B^{\leq j-1}$ and $\BSS^j$  that satisfies
the success conditions $\SSE(\leq j)$ and $\PD(\leq j-1)$. 
Let $\cJ$ be a random variable  whose value is a subset of  
$\Intervals(w_j,\epssub{i+3})$
with the property that given the fixed randomness $B^{\leq j-1}$ and $\BSS^j$, 
each $J \in \Intervals(w_j,\epssub{i+3})$ belongs to $\cJ$ with probability
at most $p$. Then
the expected
running time of $\Enumerate(j,I \times \cJ,i)$ over the choice of $\cJ$ is at most:
\begin{equation}
\label{eqn:Enumerate}
\widetilde{O}(\frac{p}{\theta} \sum_{h=1}^j d_{h-1}w_h^{1+1/T'}).
\end{equation}
\end{lemma}

\begin{proof}
The proof is by induction on $j$.  
Suppose $j=1$. We run  $\ED(z_I,z_J,\kappa)$ for each $J \in \cJ$.
The expected time is $\widetilde{O}(\frac{1}{\theta}p d_0w_1^{1+1/T'})$ since 
the expected size of $\cJ$ is at most $8p \frac{n}{\theta w_1} \le \frac{16p}{\theta} \sqrt{n} \le \frac{32p}{\theta} d_0$
and each call of $\ED$ costs $w_1^{1+1/T'}$.

Now suppose $j>1$.  The loops on $i'$ and $I'$ starting in lines (11-12) are executed $\widetilde{O}(1)$ times.
The construction of $\cJ'$ in line (13) using $\ZoomIn$ 
takes $\widetilde{O}(|\cJ'|)$ time (sort $\cJ$ in the natural order and build $\cJ'$ "from left to right").
By Proposition~\ref{prop:33}, for each $J' \in  \Intervals(w_{j-1},\epssub{i'+3})$, the number 
of $\epssub{i+3}$-aligned $w_j$-intervals $J$ such that $J' \in \ZoomIn(j,I \times J,i,I',i')$ is at most 33.   Since $\cJ$ is selected according to a probability
distribution so that no set $J$ belongs to  $\cJ$ with probability more than $p$, $\cJ'$ is sampled according to some distribution where for each 
$J' \in  \Intervals(w_{j-1},\epssub{i'+3})$ the probability of $J' \in \cJ'$ is at most $33p$.
Hence, the expected size of $\cJ'$ is at most $33 p \frac{n}{w_{j-1} \epssub{i'+3}} \leq O(p \frac{w_j}{\theta})$,
since $w_j \geq w_{j-1} \geq \pwrround{\sqrt{n}}$ and $\epssub{i'+3} \ge \theta/8$.  
This is dominated by the summand
for $h=j$ in (\ref{eqn:Enumerate}), which   
is at least 
$\frac{p}{\theta} d_{j-1}w_j^{1+1/T'}$.

By induction hypothesis, the recursive call to $\Enumerate$ in 
line (14) takes expected time $\widetilde{O}(\frac{33p}{\theta}\sum_{1 \leq h \leq j-1}d_{h-1}w_h^{1+1/T'})$ which is
$\widetilde{O}(\frac{p}{\theta} \sum_{1 \leq h \leq j-1}d_{h-1}w_h^{1+1/T'})$.

The final loop (22-26) on $J \in \cK$ requires $O(|\cK| w_j^{1+1/T'})$ time. So we need to  bound the size of $\cK$. 
$\cK$ is created in the loop on $i',I'$.  As noted
there are $\widetilde{O}(1)$ iterations of these loops, so it suffices to bound the number of elements added to $\cK$ for
a single choice of $I',i'$.  During lines (15-17), for each $J \in \cJ$, $J$
is added to $\cK$ if there is a $J' \in \cS$ that is in $\ZoomIn(i,I\times J,I',i')$.  By Proposition~\ref{prop:33},
each $J' \in \cS'$ is responsible for the addition of at most 33 intervals to $\cK$, so $|\cK| \leq 33|\cS'|$.    Now, $\cS'$ is the output of a call to $\Enumerate(j-1,I' \times \cJ',i')$ where $I' \in \SparseSample(j,I,i')$.  By the success condition for 
iteration $j-1$ of $\ProcessDense$ (Section~\ref{subsec:randomness})  there are at most $2d_{j-1}$ intervals $J' \in \Intervals(w_{j-1},\epssub{i'+3})$
classified as \CLOSE{} for $i'$. As observed in the previous paragraph, 
each of these at most $2d_{j-1}$ intervals belongs to $\cJ'$ with probability at most $33p$. So the expected size of $|\cS'| \leq 66 p d_{j-1}$.  
Thus the expected cost of the loop (22-26) is  $\widetilde{O}(p d_{j-1}w_j^{1+1/T'})$.
Combining with the other loop gives the claimed time bound for $\Enumerate$.
\end{proof}

Now we analyze the running time of $\Preprocess(j)$. 
There are $\widetilde{O}(n/w_j)$ pairs $(I,i)$ that are enumerated in the two outer loops.  
For each such pair, we construct $\Sparse(j,I,i)$ (which takes $\widetilde{O}(1)$ time), and $\Bbelow(j,I,i)$
whose running time is $\widetilde{O}(1)$ if $j=1$ and is $\widetilde{O}(w_j+|\cR(j-1,I)|+|\Intervals(w_j,\epssub{i+3})|)$ for $j>1$, which is the time to run $\APM$.  Summing over $O(\log(n))$ values of $i$ and noting that $\epssub{i} \geq \theta$, we obtain the upper bound
$\widetilde{O}(w_j+|\cR(j-1,I)|+\frac{n}{\theta w_{j-1}})$. 
Summing over $I$ gives $\widetilde{O}(n+|\cR(j-1)|+\frac{n^2}{\theta w_{j-1}w_j})$. 
$|\cR(j-1)|$ is at most the number of level $j-1$ candidates $\candidate{i}{I' \times J'}$ which is
at most $\widetilde{O}(\frac{n^2}{\theta w_{j-1}^2})$.  Since $w_h \geq \pwrround{\sqrt{n}}$ for all $h$ by assumption,
the overall time for $\Preprocess(j)$ is $\widetilde{O}(n/\theta)$.  We observe that this term is dominated
by the $h=j$ term in the inner sum of (\ref{eqn:time bound}) which is $\frac{n}{\theta^2}w_j^{1/T'}\frac{d_{j-1}}{d_j} \geq \frac{n}{\theta}$.
The asymptotics of the running time does not depend on the choice of random bits $\BSS^j$. 

We now analyze the time of $\ProcessDense(j)$. 
We will condition the analysis on fixing the random bits $(B^{\leq j-1},\BSS^j)=\beta^{*j}$ so that $\SG \wedge \SSE(\leq j) \wedge \PD(\leq j-1)$ holds.
The multiplicative cost of the outer iteration on $i$ is absorbed in the $\widetilde{O}$ term.
The main part is the while loop (lines 8-22) on $I \in \cT$.  This cost is divided into two parts, the
call to $\Enumerate$ within line (11), and the cost of  (lines 14-20) which is only executed within the "else".  

To bound the cost of the call to $\Enumerate$ in line (11), we want to apply
Lemma~\ref{lem:enumerate time}.  For the hypothesis of this lemma
we need an upper bound $p'$ on the probability of any particular $w_j$
interval being selected for $\cS$.  According to the code of $\Enumerate$,
every interval is placed in $\cS$ with probability at most 
$p=\min(1,c_0\log n/d_j)$.  However we need to consider the probability
of a given interval being placed in $\cS$ {\em conditioned on the
event $\PD(j)$}, and this can be bounded above by $p/\Pr[\PD(j)]$. 
As noted in Section~\ref{subsec:randomness}, 
$\PD(j)$ occurs with probability at least $1-n^{-9} \geq 1/2$
so we can bound the conditional probability of any interval being placed in 
$\cS$ by $2p$.
Applying Lemma~\ref{lem:enumerate time}, 
the expected time for the call to $\Enumerate$ in line (11) 
is
$\widetilde{O}(\frac{2p}{\theta}\sum_{1 \leq h \leq j} d_{h-1}w_h^{1+1/T'})$ 
which is $\widetilde{O}(\frac{1}{\theta d_j}\sum_{1 \leq h \leq j} d_{h-1}w_h^{1+1/T'})$.
The number of times this is executed is the number of possible $I$, which is
at most $|\Intervals(w_j)|=n/w_j$, so the overall expected cost of calls to $\Enumerate$ in line (11) is 
$\widetilde{O}(\frac{n}{\theta w_jd_j}\sum_{h=1}^ j d_{h-1}w_h^{1+1/T'})$,
as claimed in the theorem.

The time for executing (14-20) is dominated by the time of 
the two calls to $\Enumerate$, which are bounded to be at most
$\widetilde{O}(\frac{1}{\theta}(\sum_{h=1}^j d_{h-1}w_h^{1+1/T'})$ 
using Lemma~\ref{lem:enumerate time}
with
the trivial setting $p=1$.
The number of times this is executed is bounded by the number of times in the loop on $I$ that
$I$ is declared dense and used as a pivot.  We claim that if $\PD(j)$ holds then the number of pivots is upper bounded
by  $O(\frac{n}{\theta w_jd_j})$.   To
see this, first note that if $I$ is chosen as a pivot then  by Section~\ref{subsec:randomness}, conditioning on $\PD(j)$ implies
$|\Enumerate(j,I \times \Intervals(w_j,\epssub{i+3}),i)| \geq d_j/2$.
Furthermore, we claim that if $I$ and $I'$ are both pivots then $\Enumerate(j,I \times \Intervals(w_j,\epssub{i+3}),i)$
is disjoint from $\Enumerate(j,I' \times \Intervals(w_j,\epssub{i+3}),i)$.  Suppose  for contradiction 
that both are pivots and there is a $J$ in both sets, and that $I$ is  selected first as a pivot.
Then by the Soundness of $\Enumerate$, $\ncost(I\times J) \leq \epssub{i-\qual_{j-1}-6}$ and $\ncost(I' \times J)
\leq \epssub{i-\qual_{j-1}-6}$ and so by the triangle inequality $\ncost(I \times I') \leq \epssub{i-\qual_{j-1}-7}=\epssub{h_1}$
(where $h_1$ is defined in the pseudocode of $\Enumerate$.)  But, in that case, the pseudocode of $\Enumerate$ ensures that $I'$
is placed in $\cX$ in line (16) and therefore removed from $\cT$ in line (20), making it impossible for $I'$ to be chosen as a pivot.

Since the sets $\Enumerate(j,I \times \Intervals(w_j,\epssub{i+3}),i)$ corresponding to pivots are pairwise disjoint
subsets of $\Intervals(w_j,\epssub{i+3})$ each have size at least $d_j/2$, and $|\Intervals(w_j,\epssub{i+3})| = O(\frac{n}{\theta w_j})$,
the number of pivots is at $O(\frac{n}{\theta w_j d_j})$. Multiplying this by the cost of a single loop as bounded above,
the result is bounded above as claimed in the theorem.
\end{proof}

\subsection{Choosing the parameters}
\label{subsec:parameter choice}

The time analysis is expressed in terms of the parameters $w_1,\ldots,w_k$ and $d_0,\ldots,d_k$.
In this section we determine values of the parameters that achieve the claimed time bound.  
It is convenient to introduce parameters $\gamma_1,\ldots,\gamma_k$, $\delta_0,\ldots,\delta_k$ and $\tau$, with
$w_i=\pwrround{n^{\gamma_i}}$ and $d_i = \pwrround{n^{\delta_i}}$ and $\theta = \pwrround{n^{-\tau}}$.

Recall that the parameters of $\gapcondition$ include $\zeta>0$ and we only need our gap
algorithm to work for $\tau \leq \zeta$.   In the theorem we are allowed to choose $\zeta$ to be any positive
constant.    In the derivation below, we will see that we will need an upper bound on $\tau$
as a function of $T'$ which will be used to determine $\zeta$ in 
 the final proof of Theorem~\ref{thm:speedup} in the next section.

We impose the following conditions.
\begin{itemize}
\item
$d_0=w_1=\pwrround{\sqrt{n}}$, so $\delta_0=\gamma_1=1/2$
\item
$d_k=1$, so $\delta_k=0$.
\end{itemize}

The time for iteration $j$ is:

\[
\chi_j=\widetilde{O}\left(\frac{n}{\theta^2 w_jd_j} ( \sum_{i=1}^j d_{i-1} w_i^{\frac{T'+1}{T'}})\right).
\]

Define  
\begin{itemize}
\item
$\alpha_j = (1-\gamma_j-\delta_j+2\tau)$
\item
$\nu_i= \delta_{i-1}+(\frac{T'+1}{T'})\gamma_i$

\end{itemize}

Then the cost of processing level $j$ can be rewritten as:

\[
\chi_j=\widetilde{O}(\sum_{i=1}^j n^{\alpha_j+\nu_i}).
\]

We now choose $\gamma_i$ and $\delta_i$ subject to
the following conditions:

\begin{itemize}
\item $\gamma_1=\delta_0=1/2$
\item $\alpha_j$ is the same for all $j$
\item $\nu_i$ is the same for all $i$.
\item $\delta_k=0$.
\end{itemize}

It is easy to check that for any $B\ge 0$, the first three conditions are satisfied by:

\begin{eqnarray*}
\gamma_i & = & 1/2 +  B\frac{T'}{T'+1} - B \left(\frac{T'}{T'+1}\right)^i\\
\delta_i & = & 1/2 -B +B\left(\frac{T'}{T'+1}\right)^i
\end{eqnarray*}

The condition $\delta_k=0$  implies:

\[
B=B_k(T')= \frac{(T'+1)^k}{2((T'+1)^k-T'^k)}
\]
Then $\alpha_j=1-\gamma_j-\delta_j +2\tau= \frac{B}{T'+1}+2\tau$ and
$\nu_i = 1+\frac{1}{2T'}$.
So the time for all iterations is:

\begin{eqnarray*}
\sum_{j=1}^k \chi_j& = &\widetilde{O}\left( \sum_{j=1}^k j n^{1+\frac{1}{2T'}+\frac{B}{T'+1}+2\tau}\right)\\
& = & \frac{k(k+1)}{2}\widetilde{O}\left(n^{1+\frac{1}{2T'}+\frac{B}{T'+1}+2\tau}\right).
\end{eqnarray*}

As indicated earlier, we will impose the condition $\tau \leq \frac{3T'-2}{6(6(T')^3+7(T')^2+T')}$

For fixed $T' \geq 1$, $B_k(T')$  is a decreasing function of $k$ whose limiting value is $1/2$.   So we choose 
$k=k(T')$ to be large enough so that
$B \leq \frac{3T'+1}{6T'}$.   While the value $k(T')$ is not important, it is straightforward to verify that we can choose $k(T')=\lceil (T'+1)(1+\ln(T'+1)) \rceil$.

Using the above choice for $B$, the exponent of $n$ is at most 
$1+\frac{1}{2T'}+\frac{3T'+1}{6T'(T'+1)}+2\tau$ and a computation shows that  setting $T=T'+1/6$ 
and imposing $\tau \leq \frac{3T'-2}{6(6(T')^3+7(T')^2+T')}$  (which we can do since $T' \geq 1$) results in an upper bound
on the exponent of  $1+1/T$ as required.

Finally, we need to verify the assumptions (\ref{eqn:tech assump2}) and (\ref{eqn:tech assump}) that  $\frac{n}{w_j} \ge d_j$ and $\frac{w_j}{w_{j+1}}\leq \theta/2$.
The former is immediate as $\gamma_j+\delta_j \le 1$. For the latter, 
letting $M'=-\frac{1}{\log(n)} \log(\max_j 2w_j/w_{j+1})$, we  require that $\theta \geq n^{-M'}$, which we can
ensure for $n$ large enough by choosing $\zeta < M$, where $M=\min_j \gamma_{j+1} - \gamma_j$.

\subsection{Tying up the proof of Theorem~\ref{thm:speedup}}
\label{subsec:tie up}
We have that $\ED$ is a gap algorithm for edit distance satisfying $\gapcondition(T',\zeta',Q')$ where $T' \geq 1$,
$\zeta'>0$ and $Q' \geq 1$. We have shown  $\MAIN$
(using $\ED$ as a subroutine) that satisfies $\gapcondition(T,\zeta,Q)$ with $T=T'+1/6$ and 
 $\zeta>0$ and $Q\geq 1$ are suitably chosen (depending only on $T'$,$\zeta'$ and $Q'$.
In Section~\ref{subsec:main proof} we proved that $\MAIN$ has quality $Q=2^{\qual_k+6}$.
In section~\ref{subsec:parameter choice} we adjusted the parameters so that the running time
computed in Section~\ref{subsec:time analysis} is $\widetilde{O}(n^{1+1/T})$ provided that
$\theta \geq n^{-\frac{3T'-2}{6(6(T')^3+7(T')^2+T')}}$, $\theta \geq n^{-M/2}$ (where $M$ is
defined in Section~\ref{subsec:parameter choice}) and also $\theta \geq \zeta'$.
So we set $\zeta=\min(\zeta',M/2,\frac{3T'-2}{6(6(T')^3+7(T')^2+T')})$.

\section{Proof of Theorem~\ref{thm:main}}
\label{sec:main proof}
Here we present the (routine) construction of the algorithm $\FAED^T$ promised by Theorem~\ref{thm:main} 
Given $T$, let $\zeta(T)$ and $Q(T)$ be given by Theorem~\ref{thm:GUB}. 
 
On input $x,y$, $\FAED^T$ defines  $i_{\max}=\lfloor \zeta \log n \rfloor$ and for $i$ from 1 to $i_{\max}$, 
runs  $\MAIN$ on  input $(x,y,\theta=2^{-i},\delta=1/\zeta n\log(n))$. Define $i^*=0$ if none of the runs accepts, and otherwise define $i^*$ to be the largest index for which
run $i^*$ accepts.  $\FAED^T$ outputs $Q2^{-i^*}n$.
This is  an upper bound on $\editd(x,y)$ since  if $i^*=0$ then the output is $Qn \geq n$, and otherwise
the first requirement of $\gapcondition$ ensures that 
$\editd(x,y) \leq Q2^{-i^*}n$.
 
 We claim that for $R=2Q$, the probability that  the output exceeds $R(\editd(x,y)+n^{1-\zeta})$ 
is at most $1/n$.   If $i^*=i_{\max}$ then the output is $2Qn^{1-\zeta} \leq R (\editd(x,y)+n^{1-\zeta})$.
So  assume $i^*<i_{\max}$.
Say that the $i$th run of $\FAED^T$ {\em fails} if $\editd(x,y) \leq 2^{-i}n$ and the algorithm rejects.  The probability that some iteration fails is at most $\delta \zeta \log n \leq 1/n$
so the probability that no iteration fails is at least $1-1/n$.  If no iteration fails then in particular iteration $i^*+1$
does not fail, and since it rejects (by the choice of $i^*$) we conclude that $\editd(x,y) > 2^{-1-i^*}n$
and so $Q2^{-i^*}n \leq R \editd(x,y)\leq R(\editd(x,y)+n^{1-\zeta})$, and so $\FAED^T$ has all of the required properties.

\section{Approximate Pattern Matching}
\label{sec:apm}

In this section we descibe the implementation of the function
$\APM(I \times \cJ, \epsilon, \mathcal{R})$ from Section~\ref{subsec:primitives}.
This is a synthesis of algorithms from \cite{CDK-arxiv,CDGKS-arxiv}.

We assume that $\cR$ contains certified boxes and all $J\in \cJ$ are of the same width $\mu(I)$.

Let $\max(\cJ)=\{max(J):J \in \cJ\}$ and $\min(\cJ)=\{min(J):J \in \cJ\}$.
Let $\mathcal{R}^+$ be $\mathcal{R}$ augmented by auxiliary shortcut edges of cost 0
from $(\min(I),0)$ to $(\min(I),m)$  for all $m \in \min(\cJ)$.  
Also for $J \in \cJ$ let $J^0$ denote the interval $\{0,\ldots,\max(J)\}$.  The following was observed in~\cite{CDK-arxiv}:

\begin{proposition}
\label{prop:R+}
For all $J \in \cJ$, $\cost_{\cR}(I \times J)$ satisfies $\cost_{\cR^+}(I \times J^0) \leq \cost_{\cR}(I \times J)$
and $\cost(I \times J) \leq 2\cost_{\cR^+}(I \times J^0)$.  
\end{proposition} 

\begin{proof}
For the first inequality consider a min-cost traversal $\tau$ of $I \times J$ in the shortcut graph $\widetilde{G}(\cR)$.
We construct a traversal $\tau'$ of $I \times J^0$ of cost at most $\cost_{\cR}(\tau)$.  Consider the first shortcut edge $e=(i,j)\rightarrow (i',j')$ of $\tau$.
We may assume that prior to $e$, the path consists of a (possibly empty) sequence of horizontal edges followed by a (possibly empty) 
sequence of vertical edges.  The final such horizontal edge ends at $(\min(I),j)$ and $j \in \min(\cJ)$ so in $\widetilde{G}(\cR^+)$ we can replace
the horizontal path by the shortcut edge $(\min(I),0) \rightarrow \min(I),j)$ of cost 0 to get a path that is no more costly.

For the second inequality, consider a min-cost traversal $\rho$ of $I \times J^{0}$ in $\widetilde{G}(\cR^+)$. Let $j=0$
if the path does not use one of the auxiliary shortcut edges, and otherwise let $j$ be such that the path
starts with auxiliary shortcut edge $(\min(I),0) \rightarrow (\min(I),j)$.  Let $\hat{J}=\{j,\ldots,\max{J}\}$.
So the remaining portion of $\rho$ is a min-cost traversal  $\hat{\rho}$ of $I \times \hat{J}$.   
Since $\widetilde{G}(\cR)$ is certified, $|\mu(\hat{J})-\mu(I)| \le \cost(I \times \hat{J}) \le \cost_{\cR}(\hat{\rho}) 
= \cost_{\cR^+}(\hat{\rho})=\cost_{\cR^+}(I \times J^0)$.
Also $|J \Delta \hat{J}|=|\mu(\hat{J})-\mu(I)| = |\min(J)-j| \le \cost_{\cR^+}(I \times J^0)$.
So $\cost(I \times J) \leq \cost(I \times \hat{J})+|J \Delta \hat{J}|\le 2\cost_{\cR^+}(I \times J^0)$.
\end{proof}

So if we compute $\cost_{\cR^+}(I \times J^0)$ for every $J \in \cJ$, and output the set of all $J$ for which this cost
is less than $\kappa \mu(I)$, we will satisfy the requirements of $\APM$.  We now describe a slightly modified version
of an algorithm from~\cite{CDGKS-arxiv} that accomplishes this in time $\widetilde{O}(|\cR^+|)$.

Let $\Ht$ be the graph $\widetilde{G}(\cR^+)$ with each
cost $c_e$ of $e=(i,j)\to (i',j')$ replaced by {\em benefit} $b_e=(i'-i)+(j'-j)-c_e$, (so H and V edges have
benefit 0).
For any interval $B$, the min-cost traversal of $I \times B$ in  $\widetilde{G}(\cR^+)$ is $\mu(I)+\mu(B)$ minus the max-benefit traversal of $I \times B$ 
in $\Ht$.   So it suffices to compute  the max-benefit traversal of $I \times J^0$ in $\Ht$ for all $J \in \cJ$.

To do this, let $j_1<\cdots<j_r$ be the distinct second coordinates of the heads and tails of shortcut edges in $\widetilde{G}(\cR^+)$.  
We use a binary tree data structure with leaves corresponding to the indices of $I$,
where each tree node $v$ stores a number $a_v$, and a collection of lists $L_1$,\ldots,$L_r$, where $L_h$ stores pairs $(e,q(e))$ where the head of $e$ has $y$-coordinate $j_h$ and $q(e)$ is the max benefit of a path from $(\min(I),0)$ that ends with $e$.

We proceed in rounds $h=1,\dots,r$.  In round $h$, let $A_h$ consist of all the shortcuts whose tail has vertical coordinate $j_h$. The preconditions for round $h$ are: (1)  for each leaf $i$, the stored value $a_i$ is the max benefit path to $(i,j_h)$ that includes a shortcut 
whose head has horizontal coordinate $i$ (or 0 if there is no such path), (2) for each internal node $v$, $a_v=\max\{a_i:i \text{ is a leaf in the subtree of $v$}\}$,
and (3) for every shortcut edge $e=(i',j_{h'})\to (i'',j_{h''})$ with $h'<h$, the value $q(e)$
has been computed and $(e,q(e))$ is in list $L_{h''}$.

During round $h$,
for each shortcut $e=(i,j_h)\to(i',j_{h'})$ in $A_h$, $q(e)$ equals
the max of $a_{\ell}+b_e$ over tree leaves $\ell$ with  $\ell \leq i$.  This can be computed in $O(\log n)$ time as max $a_v+b_e$, 
where $v$ ranges over the union of $\{i\}$ with
the set of left children of vertices on the root-to-$i$ path
that are not themselves on the path. 
Add $(e,q(e))$ to list $L_{h'}$. After processing $A_h$,
update the binary tree: for each $(e,q(e)) \in L_{h+1}$, let $i$ be the horizontal coordinate
of the head of $e$ and for all vertices $v$ on the root-to-$i$ path, replace $a_v$ by $\max(a_v,q(e))$.
The tree then satisfies the precondition for round $h+1$.

To obtain the output to $\APM$,
for each $J \in \cJ$, let $h(J)$ be the index of the last iteration for which $j_{h(J)} \leq \max(J)$.
The benefit of $I \times J^0$ is the value, at the end of iteration of $h(J)$  of $a_{v_0}$ where $v_0$ is the root.

For the runtime analysis:
It would take $\widetilde{O}(\mu(I))$ time to set up the full tree data structure so we will build it incrementally by expanding only the parts of the data 
structure that contain non-zero values. Hence, the set up cost of the data structure is $O(1)$. It takes $O(|\cR^+|\log |\cR^+|)$ time to sort the shortcuts, and
$O(\log \mu(I))$ processing time per shortcut (computing $q(e)$ and later updating the data structure), overall
giving runtime $\widetilde{O}(|\cR^+|+|\cJ|)$.

\bibliography{editDistance}

\begin{thebibliography}{10}

\bibitem{AB17}
Amir Abboud and Arturs Backurs.
\newblock Towards hardness of approximation for polynomial time problems.
\newblock In {\em 8th Innovations in Theoretical Computer Science Conference,
  {ITCS} 2017, January 9-11, 2017, Berkeley, CA, {USA}}, pages 11:1--11:26,
  2017.

\bibitem{ABW15}
Amir Abboud, Arturs Backurs, and Virginia~Vassilevska Williams.
\newblock Tight hardness results for {LCS} and other sequence similarity
  measures.
\newblock In {\em {IEEE} 56th Annual Symposium on Foundations of Computer
  Science, {FOCS} 2015, Berkeley, CA, USA, 17-20 October, 2015}, pages 59--78,
  2015.

\bibitem{AHWW16}
Amir Abboud, Thomas~Dueholm Hansen, Virginia~Vassilevska Williams, and Ryan
  Williams.
\newblock Simulating branching programs with edit distance and friends: or: a
  polylog shaved is a lower bound made.
\newblock In {\em Proceedings of the 48th Annual {ACM} {SIGACT} Symposium on
  Theory of Computing, {STOC} 2016, Cambridge, MA, USA, June 18-21, 2016},
  pages 375--388, 2016.

\bibitem{And18}
Alex Andoni.
\newblock Simpler constant-factor approximation to edit distance problems.
\newblock {\em Manuscript}, 2018.

\bibitem{AKO10}
Alexandr Andoni, Robert Krauthgamer, and Krzysztof Onak.
\newblock Polylogarithmic approximation for edit distance and the asymmetric
  query complexity.
\newblock In {\em 51th Annual {IEEE} Symposium on Foundations of Computer
  Science, {FOCS} 2010, October 23-26, 2010, Las Vegas, Nevada, {USA}}, pages
  377--386, 2010.

\bibitem{AN10}
Alexandr Andoni and Huy~L. Nguyen.
\newblock Near-optimal sublinear time algorithms for {Ulam} distance.
\newblock In {\em Proceedings of the Twenty-First Annual {ACM-SIAM} Symposium
  on Discrete Algorithms, {SODA} 2010, Austin, Texas, USA, January 17-19,
  2010}, pages 76--86, 2010.
\newblock \href {http://dx.doi.org/10.1137/1.9781611973075.8}
  {\path{doi:10.1137/1.9781611973075.8}}.

\bibitem{AO09}
Alexandr Andoni and Krzysztof Onak.
\newblock Approximating edit distance in near-linear time.
\newblock In {\em Proceedings of the Forty-first Annual ACM Symposium on Theory
  of Computing}, STOC '09, pages 199--204, New York, NY, USA, 2009. ACM.

\bibitem{BI15}
Arturs Backurs and Piotr Indyk.
\newblock Edit distance cannot be computed in strongly subquadratic time
  (unless {SETH} is false).
\newblock In {\em Proceedings of the Forty-Seventh Annual ACM on Symposium on
  Theory of Computing}, STOC '15, pages 51--58, New York, NY, USA, 2015. ACM.

\bibitem{BJKK04}
Z.~Bar-Yossef, T.S. Jayram, R.~Krauthgamer, and R.~Kumar.
\newblock Approximating edit distance efficiently.
\newblock In {\em Foundations of Computer Science, 2004. Proceedings. 45th
  Annual IEEE Symposium on}, pages 550--559, Oct 2004.

\bibitem{BEKMRRS03}
Tugkan Batu, Funda Erg\"{u}n, Joe Kilian, Avner Magen, Sofya Raskhodnikova,
  Ronitt Rubinfeld, and Rahul Sami.
\newblock A sublinear algorithm for weakly approximating edit distance.
\newblock In {\em Proceedings of the Thirty-fifth Annual ACM Symposium on
  Theory of Computing}, STOC '03, pages 316--324, New York, NY, USA, 2003. ACM.

\bibitem{BES06}
Tu\u{g}kan Batu, Funda Ergun, and Cenk Sahinalp.
\newblock Oblivious string embeddings and edit distance approximations.
\newblock In {\em Proceedings of the Seventeenth Annual ACM-SIAM Symposium on
  Discrete Algorithm}, SODA '06, pages 792--801, Philadelphia, PA, USA, 2006.
  Society for Industrial and Applied Mathematics.

\bibitem{BEGHS18}
Mahdi Boroujeni, Soheil Ehsani, Mohammad Ghodsi, Mohammad~Taghi Hajiaghayi, and
  Saeed Seddighin.
\newblock Approximating edit distance in truly subquadratic time: Quantum and
  {MapReduce}.
\newblock In {\em Proceedings of the Twenty-Ninth Annual {ACM-SIAM} Symposium
  on Discrete Algorithms, {SODA} 2018, New Orleans, LA, USA, January 7-10,
  2018}, pages 1170--1189, 2018.

\bibitem{BEGHS18A}
Mahdi Boroujeni, Soheil Ehsani, Mohammad Ghodsi, Mohammad~Taghi Hajiaghayi, and
  Saeed Seddighin.
\newblock Approximating edit distance in truly subquadratic time: Quantum and
  {MapReduce} (extended version of~\cite{BEGHS18}).
\newblock 2018.

\bibitem{BR19}
Joshua Brakensiek and Aviad Rubinstein.
\newblock Constant-factor approximation of near-linear edit distance in
  near-linear time.
\newblock {\em CoRR}, abs/1904.05390, 2019.

\bibitem{BK15}
Karl Bringmann and Marvin K{\"{u}}nnemann.
\newblock Quadratic conditional lower bounds for string problems and dynamic
  time warping.
\newblock In {\em {IEEE} 56th Annual Symposium on Foundations of Computer
  Science, {FOCS} 2015, Berkeley, CA, USA, 17-20 October, 2015}, pages 79--97,
  2015.

\bibitem{CDGKS18}
Diptarka Chakraborty, Debarati Das, Elazar Goldenberg, Michal Kouck{\'{y}}, and
  Michael~E. Saks.
\newblock Approximating edit distance within constant factor in truly
  sub-quadratic time.
\newblock In {\em 59th {IEEE} Annual Symposium on Foundations of Computer
  Science, {FOCS} 2018, Paris, France, October 7-9, 2018}, pages 979--990,
  2018.
\newblock URL: \url{https://doi.org/10.1109/FOCS.2018.00096}, \href
  {http://dx.doi.org/10.1109/FOCS.2018.00096}
  {\path{doi:10.1109/FOCS.2018.00096}}.

\bibitem{CDGKS-arxiv}
Diptarka Chakraborty, Debarati Das, Elazar Goldenberg, Michal Kouck{\'{y}}, and
  Michael~E. Saks.
\newblock Approximating edit distance within constant factor in truly
  sub-quadratic time.
\newblock {\em CoRR}, abs/1810.03664, 2018.
\newblock URL: \url{http://arxiv.org/abs/1810.03664}, \href
  {http://arxiv.org/abs/1810.03664} {\path{arXiv:1810.03664}}.

\bibitem{CDK-arxiv}
Diptarka Chakraborty, Debarati Das, and Michal Kouck{\'{y}}.
\newblock Approximate online pattern matching in sub-linear time.
\newblock {\em CoRR}, abs/1810.03551, 2018.
\newblock URL: \url{http://arxiv.org/abs/1810.03551}, \href
  {http://arxiv.org/abs/1810.03551} {\path{arXiv:1810.03551}}.

\bibitem{G16}
Szymon Grabowski.
\newblock New tabulation and sparse dynamic programming based techniques for
  sequence similarity problems.
\newblock {\em Discrete Applied Mathematics}, 212:96--103, 2016.

\bibitem{LMS98}
Gad~M. Landau, Eugene~W. Myers, and Jeanette~P. Schmidt.
\newblock Incremental string comparison.
\newblock {\em SIAM J. Comput.}, 27(2):557--582, April 1998.

\bibitem{Lev65}
VI~Levenshtein.
\newblock {Binary Codes Capable of Correcting Deletions, Insertions and
  Reversals}.
\newblock {\em Soviet Physics Doklady}, 10:707, 1966.

\bibitem{MP80}
William~J. Masek and Michael~S. Paterson.
\newblock A faster algorithm computing string edit distances.
\newblock {\em Journal of Computer and System Sciences}, 20(1):18 -- 31, 1980.

\bibitem{NSS17}
Timothy Naumovitz, Michael~E. Saks, and C.~Seshadhri.
\newblock Accurate and nearly optimal sublinear approximations to {Ulam}
  distance.
\newblock In {\em Proceedings of the Twenty-Eighth Annual {ACM-SIAM} Symposium
  on Discrete Algorithms, {SODA} 2017, Barcelona, Spain, Hotel Porta Fira,
  January 16-19}, pages 2012--2031, 2017.

\bibitem{UKK85}
Esko Ukkonen.
\newblock Algorithms for approximate string matching.
\newblock {\em Inf. Control}, 64(1-3):100--118, March 1985.

\bibitem{WF74}
Robert~A. Wagner and Michael~J. Fischer.
\newblock The string-to-string correction problem.
\newblock {\em J. ACM}, 21(1):168--173, January 1974.

\end{thebibliography}

\end{document}